\documentclass[reqno]{amsart}

\usepackage{amssymb,amsmath,amscd,amsthm}

\usepackage{mathrsfs}
\usepackage{latexsym}

\usepackage{hyperref}

\newtheoremstyle{myplain}{}{}{\it}
{0pt}{\scshape}{}{ }{\thmname{#1}\thmnumber{ #2}\thmnote{ (#3)}}
\newtheoremstyle{mydefinition}{}{}{}
{0pt}{\scshape}{}{ }{\thmname{#1}\thmnumber{ #2}\thmnote{ (#3)}}

\theoremstyle{myplain}

    \newtheorem{Def}{Definition}[section]
        \newtheorem{Lem}[Def]{Lemma}
            \newtheorem{theo}[Def]{Theorem}
        \newtheorem{prop}[Def]{Proposition}
        \newtheorem{rem}[Def]{Remark}
        \newtheorem{hyp}[Def]{Hypothesis}

\renewcommand{\qed}{\nopagebreak\hfill$\Box$}

\newcommand{\skp}[2]{\mbox{$\left\langle #1\, , \, #2\right\rangle$}}
\newcommand{\skpd}[2]{\mbox{$\left\langle #1\, ,\,#2\right\rangle_{\ell^2}$}}

\newcommand{\skpR}[2]{\mbox{$\left\langle #1\, ,\,#2\right\rangle_{L^2}$}}
\newcommand{\skpk}[2]{\mbox{$\left\langle #1\, ,\,#2\right\rangle_{\mathcal V}$}}
\newcommand{\skpH}[2]{\mbox{$\left\langle #1\, ,\,#2\right\rangle_{{\mathscr H}_{\tilde{\varphi}}}$}}

\newcommand{\natop}[2]{\genfrac{}{}{0pt}{}{#1}{#2}}

\newcommand{\dnd}[2]{\mbox{$\frac{\partial #1}{\partial #2}$}}

\DeclareMathOperator{\Span}{span}

\DeclareMathOperator{\rank}{rank}
\DeclareMathOperator{\ran}{Ran}
\DeclareMathOperator{\supp}{supp}

\DeclareMathOperator{\Op}{Op}

\DeclareMathOperator{\diag}{diag}

\newcommand{\id}{\mathbf{1}}

\numberwithin{equation}{section}

\newcommand{\beqa}{\begin{eqnarray*}}
\newcommand{\eeqa}{\end{eqnarray*}}
\renewcommand{\hat}{\widehat}
\newcommand{\bauf}{\begin{itemize}}
\newcommand{\eauf}{\end{itemize}}
\newcommand{\be}{\begin{equation}}
\newcommand{\ee}{\end{equation}}
\newcommand{\ben}{\begin{enumerate}}
\newcommand{\een}{\end{enumerate}}
\newcommand{\ra}{\rightarrow}

\renewcommand{\o}{\omega}
\renewcommand{\O}{\Omega}
\newcommand{\ep}{\varepsilon}

\newcommand{\R}{{\mathbb R} }
\newcommand{\Z}{{\mathbb Z}}
\newcommand{\C}{{\mathbb C}}
\newcommand{\N}{{\mathbb N}}
\newcommand{\T}{{\mathbb T}}

\newcommand{\hN}{\tfrac{\mathbb N^*}{2}}
\newcommand{\hNnull}{\tfrac{\mathbb{N}}{2}}
\newcommand{\hZ}{\frac{\mathbb Z}{2}}
\newcommand{\disk}{(\varepsilon {\mathbb Z})^d}
\newcommand{\Hi}{\mathscr H}
\newcommand{\Ce}{\mathscr C}

\newcommand{\De}{\mathscr D}

\title{Asymptotic eigenfunctions for a class of difference operators}

\author{Markus Klein \and Elke Rosenberger}

\address{
Universit\"at Potsdam\\ Institut f\"ur
Mathematik \\ Am Neuen Palais 10\\ 14469 Potsdam}
\email{mklein@math.uni-potsdam.de, erosen@uni-potsdam.de}

\date{\today}


\begin{document}

\begin{abstract}
We analyze a general class of difference operators $H_\ep = T_\ep
+ V_\ep$ on $\ell^2(\disk)$, where $V_\ep$ is a one-well
potential and $\ep$ is a small parameter. We construct formal asymptotic expansions of
WKB-type
for eigenfunctions associated with the low lying eigenvalues of $H_\ep$. These are obtained
from eigenfunctions or quasimodes for the
operator $H_\ep$, acting on $L^2(\R^d)$, via restriction to the lattice $\disk$.
\end{abstract}

\maketitle

\section{Introduction}

The central topic of this paper is the construction of formal WKB-type expansions of
eigenfunctions for a
rather general class of families of
difference operators $\left(H_\ep\right)_{\ep\in (0,\ep_0]}$ on the Hilbert space
$\ell^2(\disk)$, as the small parameter $\ep>0$
tends to zero. The operator $H_\ep$ is given by
\begin{align} \label{Hepein}
H_\ep &= (T_\ep + V_\ep ),  \quad\text{where}\quad
T_\ep  = \sum_{\gamma\in\disk} a_\gamma \tau_\gamma ,\\
(a_\gamma\tau_\gamma u)(x) &=a_\gamma(x; \ep) u(x+\gamma)  \quad \mbox{for}
\quad x,\gamma\in\disk
\end{align}
and $V_\ep$ is a multiplication operator, which in leading order is given by
$V_0 \in \Ce^\infty (\R^d)$.

This paper is based on the thesis Rosenberger \cite{thesis}. It is the third in a series of
papers
(see Klein-Rosenberger \cite{kleinro}, \cite{kleinro2}); the aim is to develop an analytic
approach
to the semiclassical eigenvalue problem and tunneling for $H_\ep$ which is comparable
in detail and
precision
to the well known analysis for the Schr\"odinger operator (see Simon \cite{Si1},
\cite{Si2} and
Helffer-Sj\"ostrand \cite{hesjo}). Our motivation comes from
stochastic problems (see Klein-Rosenberger \cite{kleinro}, Bovier-Eckhoff-Gayrard-Klein
\cite{begk1}, \cite{begk2},
Baake-Baake-Bovier-Klein \cite{baake}). A large class of
discrete Markov chains analyzed in \cite{begk2}
with probabilistic
techniques falls into the framework of difference operators treated in this article.

In this paper we consider the case of a one-well potential $V_\ep$ and derive formal
asymptotic expansions
for the eigenfunctions $v_j$ and the associated low lying eigenvalues of $H_\ep$.
These lead to good quasimodes for $H_\ep$ (in the precise sense of Theorem \ref{theoEjaj}
below),
which will be crucial for our analysis of the tunneling problem in a subsequent paper (see
 \cite{thesis}).
In general, these expansions contain half-integer powers of $\ep$. As in \cite{hesjo}
for the case of Schr\"odinger operators, we obtain sufficient conditions for the absence
of these half-integer terms.
We approach the construction of asymptotic expansions of type
\[ v_j (x; \ep) \sim e^{-\frac{\varphi (x)}{\ep}} u_j (x; \ep) \]
by conjugation of $H_\ep$ with the exponential weight $e^{-\frac{\varphi (x)}{\ep}}$ and
subsequent
rescaling, using the variable $y=\frac{x}{\sqrt{\ep}}$. Here $\varphi$ is the Finsler distance of $x$ to the
potential well placed at $x=0$, as constructed in Klein-Rosenberger \cite{kleinro}. This leads to an operator $\hat{G}_\ep$
treated at length in Section \ref{kap22}. It is analog to the
approach in Klein-Schwarz \cite{erika} in the case of the Schr\"odinger operator. This more
elementary
approach avoids the use of an additional FBI-transform which was used in the original
WKB-analysis of
Helffer-Sj\"ostrand \cite{hesjo}. We remark that the discrete setting of the present paper
introduces numerous
technical difficulties. The main technical result in this respect is Proposition \ref{TayG}
on the
expansion of $e^{\frac{\varphi (x)}{\ep}} H_\ep e^{-\frac{\varphi (x)}{\ep}}$ (in the variable $y$).

We assume

\begin{hyp}\label{hyp1}
\begin{enumerate}
\item The coefficients $a_\gamma(x; \ep)$ in \eqref{Hepein} are
functions
\begin{equation}\label{agammafunk}
a: \disk \times \R^d \times (0,\ep_0] \ra \R\, , \qquad (\gamma, x,
\ep) \mapsto a_\gamma(x; \ep)\, ,
\end{equation}
for some $\ep_0>0$, satisfying the following conditions:
\ben
\item[(i)] They have an
expansion
\begin{equation}\label{agammaexp}
a_\gamma(x; \ep) = \sum_{k=0}^{N-1} \ep^k a_\gamma^{(k)}(x)  + R^{(N)}_\gamma (x; \ep)\, ,\qquad N\in\N\, ,
\end{equation}
where $a_\gamma \in\Ce^\infty (\R^d\times (0,\ep_0])$ and $a_\gamma^{(k)}\in\Ce^\infty(\R^d)$ for $0\leq k \leq N-1$.
\item[(ii)] $\sum_\gamma a_\gamma^{(0)}  = 0$ and $a_\gamma^{(0)}
\leq 0$ for $\gamma \neq 0$
\item[(iii)] $a_\gamma(x; \ep) =
a_{-\gamma}(x+\gamma; \ep)$ for $x \in \R^d, \gamma \in \disk$
\item[(iv)] For any $c>0$ and $\alpha\in\N^d$ there exists $C>0$ such that for
$0\leq k\leq N-1$ uniformly
with respect to $x\in\R^d$ and $\ep\in (0,\ep_0]$
\begin{equation}\label{abfallagamma}
\| \, e^{\frac{c|.|}{\ep}} \partial_x^\alpha a^{(k)}_.(x)\|_{\ell_\gamma^2(\disk)}
\leq C \qquad\text{and} \qquad
\|\,
 e^{\frac{c|.|}{\ep}} \partial_x^\alpha R^{(N)}_.(x)\|_{\ell^2_\gamma(\disk)}
 \leq C\ep^N
\end{equation}
\item[(v)]
$\Span \{\gamma\in\disk\,|\, a^{(0)}_\gamma(x) <0\}= \R^d$ for $x=0$.
\een
\item
\ben
\item[(i)] The potential energy $V_\ep$ is the restriction to $\disk$ of a
function $\hat{V}_\ep\in\Ce^\infty (\R^d, \R)$, which has an expansion
\begin{equation}\label{hatVell}
\hat{V}_{\ep}(x) = \sum_{\ell=0}^{N-1}\ep^l V_\ell(x) + R_{N}(x;\ep)  \, ,\qquad N\in\N\, ,
\end{equation}
where $V_\ell\in\Ce^\infty(\R^d)$, $R_{N}\in \Ce^\infty (\R^d\times (0,\ep_0])$ for some $\ep_0>0$ and
for any compact set $K\subset \R^d$
there exists a constant $C_K$ such that $\sup_{x\in K} |R_{N}(x;\ep)|\leq C_K \ep^{N}$.
\item[(ii)]
$V_\ep$ is polynomially bounded and there exist constants $R, C > 0$ such that
$V_\ep(x) > C$ for all $|x| \geq R$ and $\ep\in(0,\ep_0]$.
\item[(iii)]
$V_0(x)\geq 0$ and it takes the value $0$ only at the non-degenerate minimum $x_0=0$,
which we call the potential well.
\een
\een
\end{hyp}

If $\T^d := \R^d/(2\pi)\Z^d$ denotes the $d$-dimensional torus and
$b\in \Ce^\infty\left(\R^d\times \T^d\times (0,1]\right)$,
a pseudo-differential operator $\Op_\ep^{\T}(b): {\mathcal K}\left(\disk\right)
\longrightarrow
{\mathcal K}'\left(\disk\right)$ is defined by
\begin{equation}\label{psdo2}
\Op_\ep^{\T}(b)\, v(x) := (2\pi)^{-d} \sum_{y\in\disk}\int_{[-\pi,\pi]^d} e^{\frac{i}{\ep}(y-x)\xi}
b(x,\xi;\ep)v(y) \, d\xi \, ,
\end{equation}
where
\begin{equation}\label{kompaktge}
{\mathcal K}\left(\disk\right):=\{ u: \disk\rightarrow \C\; |\; u~\mbox{has compact support}\}
\end{equation}
and ${\mathcal K}'\left(\disk\right):= \{f: \disk\rightarrow \C\ \} $ is dual to
${\mathcal K}\left(\disk\right)$
by use of the scalar product $\skpd{u}{v}:= \sum_x \bar{u}(x)v(x)$ (see the appendix of
\cite{kleinro2} for the basic theory
of such operators).

We remark that under the assumptions given in Hypothesis \ref{hyp1}, one has for $T_\ep$
defined in
\eqref{Hepein}, $T_\ep = \Op_\ep^{\T}(t(.,.;\ep))$, where
$t\in\Ce^\infty\left(\R^d\times\T^d\times (0,\ep_0]\right)$ is given by
\begin{equation}\label{talsexp}
t(x,\xi; \ep) = \sum_{\gamma\in\disk} a_\gamma (x; \ep) \exp \left(-\frac{i}{\ep}\gamma
\cdot\xi\right)\; .
\end{equation}
Here $t$ is considered as a function on $\R^{2d}\times (0,\ep_0]$, which is
$2\pi$-periodic with respect to $\xi$.

Furthermore, we set
\begin{align}\label{texpand}
t(x,\xi;\ep) &= \sum_{k=0}^{N-1} \ep^k t_k (x,\xi) + \tilde{t}_N(x,\xi;\ep)\; ,\qquad
\text{with}\\
t_k(x,\xi) &:= \sum_{\gamma\in\disk} a_\gamma^{(k)}(x) e^{-\frac{i}{\ep}\gamma\xi}\, ,
\qquad 0\leq k \leq N-1\nonumber\\
\tilde{t}_N(x,\xi;\ep) &:= \sum_{\gamma\in\disk} R_\gamma^{(N)}(x; \ep) e^{-\frac{i}{\ep}
\gamma\xi}\nonumber\; .
\end{align}

Thus, in leading order, the symbol of $H_\ep$ is $h_0:=t_0+V_0$.

\setcounter{Def}{0}
\begin{hyp}
\ben
\item[(c)] We assume that $t_0$ defined in \eqref{texpand} fulfills
\[ t_0(0, \xi) >0\, , \quad\text{if} \quad |\xi| >0 \, .\]
\een
\end{hyp}

A simple example for an operator satisfying Hypothesis \ref{hyp1} is the discrete Laplacian. Here we have
$a_\gamma = -1$ if $|\gamma|=\ep$, $a_0 = 2d$ and $a_\gamma = 0$ else independend of $x$ and $\ep$, leading to
\[ t(x,\xi; \ep) = t_0(x,\xi) = 2 \sum_ {\nu = 1}^d (1-\cos \xi_\nu)\, . \]

\begin{rem}\label{rem1}
It follows from (the proof of) Klein-Rosenberger \cite{kleinro}, Lemma 1.2,
that under the assumptions given in Hypothesis \ref{hyp1}:
\ben
\item $\sup_{x,\xi}|\partial_x^\alpha\partial_\xi^\beta
t(x,\xi; \ep)|\leq C_{\alpha,\beta}$ for all $\alpha,\beta\in\N^d$
uniformly with respect to $\ep$. Moreover $t_k,\, 0\leq k\leq N-1,$ is bounded and
$\sup_{x,\xi}|\tilde{t}_N(x,\xi; \ep)| = O(\ep^N)$.
\item By Hypothesis \ref{hyp1}, (a)(i), the condition (a)(v) holds for all $x$ in a
neighborhood of $0$. This is sufficient to prove all
statements of this paper. For the more global results of \cite{kleinro}, it is necessary to
assume (v) for all $x\in\R^d$.
\item At $\xi=0$, for fixed $x\in\R^d$, the function $t_0$ defined in \eqref{texpand} has an
expansion
\begin{equation}\label{kinen}
t_0(x, \xi) = \skp{\xi}{B(x)\xi} + \sum_{\natop{|\alpha|=2n}{n\geq2}} B_\alpha (x) \xi^\alpha \qquad\text{as}\;\; |\xi|\to 0\, ,
\end{equation}
where $\alpha\in\N^d$, $B\in \Ce^\infty (\R^d, \mathcal{M}(d\times d,\R))$ is positive
definite in a neighborhood of zero,
$B(x)$ is symmetric and $B_\alpha$ are real functions.
By straightforward calculations one gets
for $\mu,\nu\in\N$
\begin{equation}\label{Bnumu}
B_{\nu\mu}(x) = -\frac{1}{2\ep^2} \sum_{\gamma\in\disk} a_\gamma^{(0)}(x)
\gamma_\nu\gamma_\mu\; .
\end{equation}
\item By Hypothesis \ref{hyp1}(a)(iii) and since the $a_\gamma$ are real,  the operator
$T_\ep$ defined in \eqref{Hepein} is symmetric.
In the probabilistic context, which is our main motivation, the former is a standard
reversibility condition while the
latter is automatic for a Markov chain.
Moreover, $T_\ep$ is bounded (uniformly in $\ep$) by condition (a)(iv) and
bounded from below by $-C\ep$ for some $C>0$ by condition (a)(iv),(iii) and (ii).
\item A combination of the expansion \eqref{agammaexp} and the reversibility condition
(a)(iii) establishes
that the $2\pi$-periodic function $\R^d\ni\xi\mapsto t_0(x,\xi)$ is even with respect to
$\xi\mapsto -\xi$, i.e.,
$a_\gamma^{(0)}(x) = a_{-\gamma}^{(0)}(x)$ for all $x\in\R^d, \gamma\in\disk$.
\item By condition (a),(iv) in Hypothesis \ref{hyp1}, the exponential decay of the
coefficients $a_\gamma$ with
respect to $\gamma$, the $2\pi$-periodic function $\R^d\ni\xi \mapsto t(x, \xi; \ep)$ has an
analytic continuation to $\C^d$. Moreover for all $B>0$
\[ \sum_\gamma \left| a_\gamma(x; \ep)\right| e^{\frac{B|\gamma|}{\ep}} \leq C \; .\]
uniformly with respect to $x$ and $\ep$.
This yields in particular
\begin{equation}\label{agammasupnorm}
\sup_{x\in\R^d} |a_\gamma(x; \ep)| \leq C e^{-\frac{B|\gamma|}{\ep}}
\end{equation}
We further remark that condition (a),(iv) implies the estimate $|a_\gamma^{(k)}(x) - a_\gamma^{(k)}(x + h)| \leq C h$ for
$0\leq k \leq N-1$ uniformly with respect to $\gamma\in \disk$ and $h,x\in \R^d$.
\item Since $T_\ep$ is bounded, $H_\ep=T_\ep + V_\ep$ defined in
\eqref{Hepein} possesses a self adjoint realization on the maximal domain of $V_\ep$.
Abusing notation, we
shall denote this realization also by $H_\ep$ and its domain by $\De(H_\ep)\subset
\ell^2\left(\disk\right)$. The associated symbol is denoted by $h(x,\xi;\ep)$.
Clearly, $H_\ep$ commutes with complex conjugation.
\een
\end{rem}

We will use the notation
\begin{equation}\label{agammaunep}
\tilde{a}:\Z^d\times\R^d\ni (\eta,x)\mapsto \tilde{a}_\eta(x) := a_{\ep\eta}^{(0)}(x)\in\R
\end{equation}
and we have by Remark \ref{rem1} (d) and \eqref{agammaunep}
\begin{equation}\label{tildehnull}
-\tilde{h}_0(x, \xi) := h_0(x,i\xi) =  \sum_{\eta\in\Z^d} \tilde{a}_\eta(x) \cosh
\left(\eta\cdot \xi\right) + V_0(x)\,:\, \R^{2d} \ra \R\; .
\end{equation}

As shown in \cite{kleinro}, the eigenfunctions associated to the first eigenvalues are
localized in a small neighborhood of the potential well.

We shall construct asymptotic expansions of WKB-type for the eigenfunctions and associated
low lying eigenvalues of
$H_\ep$. It is crucial for our approach that our actual constructions of quasimodes are done
for the operator on $L^2(\R^d)$
\begin{equation}\label{hutHep}
\hat{H}_\ep = \hat{T}_\ep +  \hat{V}_\ep\, , \qquad
\hat{T}_\ep = \sum_{\gamma\in\disk} a_\gamma(x;\ep) \tau_\gamma \; ,
\end{equation}
where $a_\gamma (x;\ep)$ and $\hat{V}_\ep$ satisfy Hypothesis \ref{hyp1}.
Alternatively, $\hat{H}_\ep = \Op_\ep (t + \hat{V}_\ep)$, where $\Op_\ep$ denotes the usual
$\ep$-dependent quantization for symbols in $S_\delta^r(m)(\R^{2d})$ (see \cite{kleinro2}) and
the symbol $t$ is defined in \eqref{talsexp}. For $\hat{H}_\ep$, it is easy to change
coordinates by
translations and rotations. One then observes that restriction of an eigenfunction (or a
quasi-mode) of
$\hat{H}_\ep$ to the lattice $\disk$ gives an eigenfunction (or a quasi-mode) of $H_\ep$ with
the same (approximate)
eigenvalue.
More precisely, for $x_0\in\R^d$, let us denote by $\mathscr{G}_{x_0}= \disk + x_0$ the
corresponding
affine lattice. Then $H_\ep$ acts in a natural way on $\ell^2(\mathscr{G}_{x_0})$, since
restriction to $\mathscr{G}_{x_0}$ and translation by $\gamma\in\disk$ commutes.
If $r_{\mathscr{G}_{x_0}}$ denotes
restriction to $\mathscr{G}_{x_0}$, we have
\begin{equation}\label{Gxnull}
 H_\ep \, r_{\mathscr{G}_{x_0}} = r_{\mathscr{G}_{x_0}}\, \hat{H}_\ep \, , \qquad
x_0\in \R^d\, .
\end{equation}
We remark that in this context the choice of $x_0 =0$ in Hypothesis \ref{hyp1},(b)(iii) is
arbitrary.
In this paper we shall systematically analyze spectrum and eigenfunctions of $H_\ep$ by
constructing quasimodes
for $\hat{H}_\ep$. We shall, however, not discuss the spectrum and the eigenfunctions of
$\hat{H}_\ep$
as an operator on $L^2(\R^d)$. This would involve a discussion of the infinite degeneracy of
each eigenvalue of
$\hat{H}_\ep$, which is not relevant in the context of this paper.

We denote the operator associated to the symbol $h_0$ by
\begin{equation}\label{Hnull}
 \hat{H}_0 u(x) := \sum_{\gamma\in\disk} a_\gamma^{(0)}(x) u(x+\gamma) + V_0(x)  u(x)  \; .
\end{equation}
Furthermore for $B$ defined in \eqref{Bnumu} we set
$B_0:= B(0)$.
The harmonic approximation of $\hat{H}_\ep$ (associated with a small
neighborhood of $(x,\xi) = (0,0)$, see \cite{kleinro2}) is given
by
\begin{equation}\label{Hnullhut}
 \hat{H}_q^0(x,\ep D) = -\ep^2 \skp{D}{B_0 D} + \skp{x}{Ax} + \ep (t_1(0,0) + V_1(0))\; ,
\end{equation}
where $A=D^2 V_0(0)$.
It follows from the assumptions given in Hypotheses \ref{hyp1} that the matrix,
$\tilde{A}:= B_0^{\frac{1}{2}}AB_0^{\frac{1}{2}}$ is symmetric and there exists an orthogonal
matrix
$R\in\text{SO}(d,\R)$ such that $R\tilde{A}R^t =\Lambda$, where
$\Lambda =\diag (\lambda_1^2, \ldots, \lambda_d^2)$ and $\lambda_\nu >0$ for $1\leq\nu\leq d$.
Therefore, by means of the unitary transformation
\begin{equation}\label{unitrans}
  Uf(x) = \sqrt{|B_0^{-\frac{1}{2}}|} f(R B_0^{-\frac{1}{2}} x) \; ,
\end{equation}
$\hat{H}^0_q$ is unitarily equivalent to the associated harmonic oscillator
\begin{equation}\label{Hharmonisch}
\hat{H}_q^{'0}(x,\ep D) := -\ep^2\Delta + \sum_{\nu =1}^d\lambda_\nu x_\nu^2 +
\ep (t_1(0,0) + V_1(0))
= U^{-1}\hat{H}^0_q U \, ,\qquad x\in\R^d \, .
\end{equation}
and we set
\begin{equation}\label{hatHstrich}
\hat{H}'_\ep:= U^{-1} \hat{H}_\ep U = \hat{T}'_\ep + \hat{V}'_\ep\; ,
\end{equation}
where, with the notation $C:= R B_0^{-\frac{1}{2}}$,
\begin{align}
\hat{V}'_\ep(x)  &= \hat{V}_\ep \left(C^{-1}x\right)  \label{hatVstrich} \\
\hat{T}'_\ep f(x)&= \sum_{\gamma\in\disk} a_\gamma(C^{-1}x; \ep) f(x + C\gamma)=
\sum_{\mu\in\Gamma_C} \left(a'_\mu\tau_\mu\right)  f(x)\; ,
\label{hatTstrich}
\end{align}
where in the last equation we set $\Gamma_C:= C\disk$ and for $\mu\in\Gamma_C$
\begin{equation}\label{astrich}
\left(a'_\mu\tau_\mu\right) f(x) = a'_\mu(x; \ep) f(x + \mu)\quad \text{with}\quad
a'_\mu(x; \ep):= a_{C^{-1}\mu}(C^{-1}x; \ep)\; .
\end{equation}
Since $C$ has maximal rank, it follows at once by direct calculation, that Hypothesis
\ref{hyp1}(a) holds for $a'$, if $\disk$ is replaced by
$\Gamma_C$ and $\gamma$ by $\mu$.
By $h_0' = t_0' + V_0'$ we denote the symbol associated to the operator $U^{-1}\hat{H}_0 U$
with respect
to the $\ep$-quantization given in \eqref{psdo2}. Then
\begin{equation}\label{tnullstrich}
 t'_0(x,\xi) = \skpd{\xi}{B'(x)\xi} + O(|\xi|^3)\, , \qquad B'(x) = \id + o(\id)
\end{equation}
and
\begin{equation}\label{Vnullstrich}
V'_{0}(x) = \sum_{\nu =1}^d\lambda_\nu^2 x_\nu^2 + O(|x|^3) =: V'_{0,q}(x) + O(|x|^3) \, ,
\quad \lambda_\nu>0\; .
\end{equation}
Moreover the symbol associated to $\hat{T}'_\ep$ is given by
\begin{equation}\label{tepstrich}
t'_\ep (x, \xi) = \sum_{\gamma\in\disk} a_\gamma(C^{-1} x; \ep) \exp
\left(-\frac{i}{\ep}C\gamma\cdot \xi\right) =\sum_{\mu\in\Gamma_C} a'_\mu( x; \ep) \exp
\left(-\frac{i}{\ep}\mu\cdot \xi\right) \; .
\end{equation}
The eigenfunctions of $\hat{H}_q^{'0}$ are given by
\begin{equation}\label{gnj}
g_{\alpha}(x) = \ep^{-\frac{d}{4}}h_{\alpha}\left(\tfrac{x}{\sqrt{\ep}}\right)
e^{-\frac{\varphi_0(x)}{\ep}} \: .
\end{equation}
where $h_{\alpha}(x)=h_{\alpha_1}(x_1)\ldots h_{\alpha_d}(x_d)$.
Each $h_{\alpha_\nu}$ is a one-dimensional Hermite polynomial,
which is assumed to be normalized in the sense that $\|g_{\alpha}\|_{L^2}=1$.
The phase function $\varphi_0$ is given by
\begin{equation}\label{phi0}
\varphi_0(x) := \sum_{\nu =1}^d \frac{\lambda_\nu }{2} x_\nu^2
\, ,\qquad x\in\R^d\, ,
\end{equation}
solving the harmonic eikonal equation $|\nabla\varphi_0(x)|^2 = V'_{0,q}(x)$.
The eigenfunctions of $\hat{H}_q^0$ are thus given by $\tilde{g}_\alpha := U g_\alpha$.

The following lemma concerns the existence of a local solution of a generalized eikonal
equation.

\begin{Lem}\label{eikonalphi}
Under the assumptions given in Hypothesis \ref{hyp1}, there exists a unique
$\Ce^\infty$-function $\varphi$ defined in a neighborhood $\Omega$ of $0$, with $\varphi (0) = 0$, solving
\begin{equation}\label{eikonal}
\tilde{h}'_0 (x, \nabla \varphi (x)):= -h'_0(x, i\nabla\varphi (x)) = 0\; ,\qquad x\in\Omega\, .
\end{equation}
Furthermore
\begin{equation}\label{varphi}
 \left|\varphi(x) - \varphi_0(x)\right| = O(|x|^3)\quad\text{as}\quad |x|\to 0\; ,
\end{equation}
and the homogeneous Taylor polynomials $\varphi_k$ of degree $k+2,\; k\geq 1$, of $\varphi$ are  constructively determined by solving
transport equations depending on the Taylor expansion of $h'_0$ at $(x,\xi) = (0,0)$.
\end{Lem}

\begin{rem} It follows from the proof of Theorem 1.5 in Klein-Rosenberger \cite{kleinro},
that $\varphi$ coincides in $\Omega$
with the Finsler distance $d^0(x)$. This proof uses Lemma \ref{eikonalphi}, which is taken
from the dissertation \cite{thesis}.
For the sake of the reader, we shall recall the proof of Lemma \ref{eikonalphi}  here.
\end{rem}

\begin{hyp}\label{tildevarphihyp}
For $\Omega, \varphi$ as in Lemma \ref{eikonalphi}, we choose a neighborhood $\O_1\subset
\Omega$ of
$0$ such that for any $\delta >0$ and for some $C>0$
the estimate $|\nabla \varphi(x)|\geq C$ holds for
$x\in\O_1\setminus \{|x|\leq \delta\}$.
We consider some set $\O_2$ such that $\overline{\O_2}\subset
\O_1$ and define a smooth cut-off function $\chi$ supported
in $\O_1$ such that $\chi(x)=1$ for any $x\in \O_2$.
Then we set for any $b>0$
\begin{equation}\label{tildevarphi}
\tilde{\varphi}(x) := \chi (x) \varphi(x) +  (1 - \chi (x)) b |x|\, ,\qquad x\in\R^d\, .
\end{equation}
\end{hyp}

The central result of this paper is the construction of the following system of quasimodes
of WKB-type, both for
the operators $\hat{H}'_\ep$ on $L^2(\R^d)$ and $H_\ep$ on $\ell^2(\disk)$.

\begin{theo}\label{theoEjaj}
Let, for $\ep>0$, $\hat{H}_{\ep}$ and $H_\ep$ respectively be an
Hamilton operator satisfying Hypotheses \ref{hyp1}.
Let $\tilde{\varphi}, \O_1$ and $\O_2$ satisfy
Hypothesis \ref{tildevarphihyp}.
Furthermore we assume that $\ep E$ denotes an eigenvalue of $\hat{H}_q^{'0}$ defined in
\eqref{Hharmonisch} with multiplicity $m$.
\ben
\item
Then there are functions $u_j\in\Ce_0^\infty\left(\R^d\times [0,\ep_0)\right), u_{j\ell}
\in \Ce_0^\infty(\R^d)\, , j=1,\ldots,m\, , \;
\ell\in\frac{\mathbb{Z}}{2}\,,\; \ell \geq -N$ for some $N$, such that for all $M\in
\frac{\mathbb{Z}}{2}$ there are
$C_M <\infty$ satisfying
\begin{equation}\label{arep}
\left| u_j(x; \ep) - \sum_{\natop{\ell\in\hZ}{\ell\geq -N}}^M  \ep^\ell u_{j\ell}(x)
\right| \leq C_M \ep^M\, , \quad (x\in\R^d)\; ,
\end{equation}
and
real functions $E_j(\ep)$ with asymptotic expansion
\begin{equation}\label{Erep}
E_j(\ep) \sim E + \sum_{k\in\frac{\N^*}{2}} \ep^k E_{jk}\; ,
\end{equation}
solving the equation
\begin{equation}\label{eigenwg}
(\hat{H}'_{\ep} - \ep E_j(\ep))\left(u_j(x,\ep
)e^{-\frac{\tilde{\varphi} (x)}{\ep}}\right) =
O\left(\ep^{\infty}\right)e^{-\frac{\tilde{\varphi} (x)}{\ep}}\,,
\quad (x\in\O_3,\;\ep \to 0)   \, ,
\end{equation}
for some neighborhood $\Omega_3 \subset \Omega_2$ of zero, where the rhs of \eqref{eigenwg} is $O(|x|^\infty)$ as $|x|\to 0$.
\item
The approximate eigenfunctions
\[ v_j:= U\left( u_j e^{-\frac{\tilde{\varphi}}{\ep}}\right)\]
of $\hat{H}_\ep$ are almost orthonormal in the sense that
\begin{equation}\label{ortho}
 \skpR{v_j}{v_k} = \delta_{jk}  + O(\ep^\infty) \; .
\end{equation}
\item
We set $I_E:=\{ \alpha\in\N^d\, |\, \hat{H}_q^{'0} g_\alpha = \ep E g_\alpha\,\}$, where
$g_\alpha$ is given  by \eqref{gnj}.
If $|\alpha|$ is even (or odd resp.) for all $\alpha \in I_E$, then all half integer terms
(or integer terms resp.) in the expansion
\eqref{arep} vanish. Moreover if $|\alpha|$ is even or odd for all $\alpha\in I_E$, the half
integer terms in \eqref{Erep} vanish.
\item
For any $x_0\in \R^d$, the restriction $v_j^\ep := r_{{\mathscr G}_{x_0}}v_j$ of the
approximate eigenfunctions to the lattice ${\mathscr G}_{x_0} =
\disk + x_0$ are approximate eigenfunctions for the operator
$H_\ep$ with respect to the approximate eigenvalues given in \eqref{Erep}, i.e.,
\begin{equation}\label{eigenwgdisk}
(H_{\ep} - \ep E_j(\ep)) v_j^\ep(x) =
O\left(\ep^{\infty}\right) Ue^{-\frac{\tilde{\varphi}}{\ep}}(x)\,,
\quad (x\in\O_3\cap {\mathscr G}_{x_0},\;\ep \to 0)   \, ,
\end{equation}
where the rhs of \eqref{eigenwgdisk} is $O(|x|^\infty)$ as $|x|\to 0$.
\item For the restricted approximate eigenfunctions we have
\begin{equation}\label{ortho2}
\skpd{v_j^\ep}{v_k^\ep} = \ep^{-d} \left( \delta_{jk} + O(\sqrt{\ep})\right)\; .
\end{equation}
\een
\end{theo}

We shall use Theorem \ref{theoEjaj} in a forthcoming paper to obtain sharp estimates on
tunneling
(see also \cite{thesis}).
The plan of the paper is as follows. Section \ref{seceic} consists of the proof of Lemma
\ref{eikonalphi}.
In Section \ref{kap22} we prove asymptotic results for the operator $\hat{G}_\ep$, which is a
unitary transform of $\hat{H}_\ep$. Here we change variables and introduce an exponential
weight. Then we use this expansion of $\hat{G}_\ep$ to define an operator $G$
on spaces of formal symbols. In Section \ref{harmapprox}, we construct asymptotic expansions
of eigenfunctions of $G$.
Section \ref{approxeigen} gives the proof of Theorem \ref{theoEjaj}.
We emphasize that the results of Sections \ref{kap22} and \ref{harmapprox} concern expansions
for
operators on spaces of formal symbols. These results are crucial for the proof of Theorem
\ref{theoEjaj}.

\section{Proof of Lemma \ref{eikonalphi}}\label{seceic}

If we formally compute the left hand side of \eqref{eigenwg}
and expand the coefficients of $e^{-\frac{\varphi }{\ep}}$ in
powers of $\ep$, the equation of order zero
determines the function $\varphi$. The order zero term of the
conjugated potential energy is $V'_0$, since $\hat{V}'_\ep$ commutes
with $e^{\frac{\varphi(x)}{\ep}}$. The conjugated kinetic term is
for $u\in L^2\left(\R^d\right)$ given by
\[
e^{\frac{\varphi}{\ep}}\hat{T}'_\ep e^{-\frac{\varphi}{\ep}} u(x) =
 \sum_{\gamma\in\Gamma_C} a'_\gamma(x; \ep)
e^{\frac{1}{\ep}(\varphi(x)-\varphi(x+\gamma))}u(x+\gamma)\; .
\]
If in addition $u\in{\mathscr C}^1\left(\R^d\right)$ and
$\varphi\in{\mathscr C}^2\left(\R^d\right)$, using the Taylor
expansion of $\varphi(x+\gamma)$ and $u(x+\gamma)$ at $x$, the
last sum is equal to
\begin{equation}\label{Tzuteik}
\sum_{\gamma\in\Gamma_C}a'_\gamma(x; \ep)
e^{\frac{1}{\ep}(-\gamma\cdot\nabla\varphi(x) -
\sum_{\nu\mu}\gamma_\nu\gamma_\mu
\int_0^1\partial_\mu\partial_\nu(\varphi(x+t\gamma))(1-t)\,dt}
\left( u(x)+ \int_0^1 \nabla u(x+t\gamma)\cdot\gamma \,dt \right)\; .
\end{equation}
The term of order
zero in $\ep$ is for $\gamma = \ep \eta$ and $\tilde{a}_\eta$ defined in \eqref{agammaunep}
\begin{equation}\label{zerot}
\sum_{\eta\in C\Z^d} \tilde{a}'_\eta(x)
e^{- \eta \cdot\nabla\varphi(x)}u(x) =
t'_0 (x,-i\nabla\varphi(x)) u(x)\; .
\end{equation}
Thus the resulting order zero part of
\eqref{eigenwg} is the generalized eikonal equation \eqref{eikonal}.

Following Helffer (\cite{helf}), the idea of the proof is to
determine $\varphi$ as generating function of a lagrangian
manifold $\Lambda_+ = \{(x,\nabla\varphi(x))\,|\,
(x,\xi)\in{\mathscr N}\}$ lying in the "energy shell"
$\left(\tilde{h}'_0\right)^{-1}(0)$, where ${\mathscr N}$ is a neighborhood of
$(0,0)$. By Hypothesis \ref{hyp1}, $\tilde{h}'_0$ expands in a
neighborhood of $(0,0)$ in $T^*\R^d$ as
\begin{equation}\label{p}
\tilde{h}'_0(x,\xi) = \skp{\xi}{B'(x)\xi} - \sum_{\nu
=1}^d\lambda_\nu^2x_\nu^2 + O\left(|\xi|^3+|x|^3\right)\; ,
\end{equation}
where $B'(0)=\id$.
Thus by the symmetry of the matrix $B'$, the Hamiltonian vector
field of $\tilde{h}'_0$ in a neighborhood of $(0,0)$ expands as
\begin{eqnarray}\label{Xtildeh0}
X_{\tilde{h}'_0} &=& 2\sum_{\nu =1}^d\left(\sum_{\mu=1}^d
B'_{\nu\mu}(x) \xi_\mu\frac{\partial}{\partial x_\nu} +
\left(\lambda^2_\nu x_\nu + \sum_{\mu,\eta=1}^d\frac{\partial
B'_{\mu\eta}}{\partial x_\nu}(x) \xi_\mu\xi_\eta
\right)\frac{\partial}{\partial \xi_\nu }\right) +
 O\left(|\xi|^2+|x|^2\right) = \nonumber\\
&=& 2\sum_{\nu =1}^d\left( \sum_{\mu=1}^d
B'_{\nu\mu}(x)\xi_\mu\frac{\partial}{\partial x_\nu} +
\lambda^2_\nu x_\nu\frac{\partial}{\partial \xi_\nu }\right) +
  O\left(|\xi|^2+|x|^2\right)\, .\label{Xq}
\end{eqnarray}
The
linearization of $X_{\tilde{h}'_0}$ at the critical point $(0,0)$
yields the fundamental matrix
\begin{equation}
L := DX_{\tilde{h}'_0}(0,0) = 2 \begin{pmatrix} 0 & \begin{matrix} 1&  & 0 \\
 & \ddots &  \\ 0 & & 1\end{matrix}  \\\begin{matrix}\lambda_1^2 &  & 0 \\
 & \ddots &  \\ 0 & & \lambda_d^2 \end{matrix} & 0 \end{pmatrix}\; .
\end{equation}
$L$ has the eigenvalues $\pm 2\lambda_\nu\, ,\nu =1,\ldots d$.
An eigenvector $(x,\xi)$ with respect to $\pm \lambda_\nu$
fulfills $\xi_\nu =\pm \lambda_\nu x_\nu$. By $\Lambda^0_{\pm}$ we
denote the positive (resp. negative) eigenspace of $L$.
$\Lambda^0_{\pm}$ can be characterized as the subsets of phase space,
which consist of all points $(x,\xi)$ such that $e^{-tL}(x,\xi)\to
0$ for $t\to \pm\infty$. Moreover, $\Lambda_{\pm}^0$ are
Lagrangian subspaces of $T_{(0,0)}(T^*\R^d)$ of the form $\xi =
\pm\nabla \varphi_0(x)$ with $\varphi_0$ defined in (\ref{phi0}).

Denote by $F_t$ the flow of the hamiltonian vector field
$X_{\tilde{h}'_0}$. By the Local Stable Manifold Theorem
(\cite{abma}), there is an open neighborhood
$\mathscr{N}$ of $(0,0)$ in $T^*\R^d$, such that
\begin{equation}
\Lambda_{\pm} :=
\left\{\left.(x,\xi)\in\mathscr{N} \,\right|\, F_t(x,\xi)\to (0,0)
\quad\text{for}\quad t\to \mp\infty\right\}
\end{equation}
are $d$-dimensional submanifolds tangent to $\Lambda_{\pm}^0$ at $(0,0)$ (the stable
($\Lambda_-$) and unstable ($\Lambda_+$) manifold of
$X_{\tilde{h}'_0}$ at the critical point $(0,0)$).
$\Lambda_+$ and $\Lambda_-$ are contained in
$\left(\tilde{h}'_0\right)^{-1}(0)$, because $\tilde{h}'_0(F_t(x,\xi)) =
\tilde{h}'_0(x,\xi)$.

In order to show that the tangent spaces at each point
$(x,\xi)\in\Lambda_{\pm}$ are Lagrangian linear subspaces of
$T_{(x,\xi)}(T^*\R^d)$, we have to show, that the
canonical symplectic form $\omega=\sum_{j=1}^d d\xi_j\wedge dx_j$
vanishes for all $u,v\in T_{(x,\xi)}(\Lambda_{\pm})$. The
Hamiltonian flow leaves the symplectic form invariant, we
therefore find for $(u,v)\in T_{(x,\xi)}(\Lambda_+)$
\[ \omega_{(x,\xi)}(u,v) = \omega_{F_t(x,\xi)}((DF_t)u,(DF_t)v)\; .\]
In the limit $t\to -\infty$, the elements of
$T_{(x,\xi)}(\Lambda_+)$ lie in the Lagrangian plane
$\Lambda_+^0$, where the symplectic form vanishes, thus
$\omega_{(x,\xi)}(u,v) = 0$ for all $(u,v)\in
T_{(x,\xi)}(\Lambda_+)$.

The projection $(x,\xi)\mapsto x$ defines a diffeomorphism of
$\mathscr{N}\cap \Lambda_+$ onto a sufficiently small neighborhood
$\Omega$ of 0 in $\R^d$. Therefore we can parameterize $\Lambda_+$
as the set of points $(x_1,\ldots x_d,\Psi_1(x),\ldots \Psi_d(x))$
with $\Psi_\nu\in \mathscr{C}^{\infty}(\Omega)$. Since $\Lambda_+$
is Lagrangian, we can deduce $\frac{\partial \Psi_\nu}{\partial
x_\mu} = \frac{\partial \Psi_\mu}{\partial x_\nu}$ and there
exists a function $\varphi\in\mathscr{C}^{\infty}(\Omega)$ with
\[ \nabla \varphi (x) = \Psi (x)\quad\textrm{and}\quad \varphi (0) = 0
\, .\] Since $T_{(0,0)}(\Lambda_{\pm})= \Lambda_{\pm}^0$, the
leading order term of this function $\varphi$ is equal to
$\varphi_0$, thus $\varphi$ can be written as (\ref{varphi}).
Furthermore $\varphi$ solves the eikonal equation (\ref{eikonal}),
because
$\Lambda_+\subset \left(\tilde{h}'_0\right)^{-1}(0)$.\\

With the ansatz (\ref{varphi}), we have a constructive procedure
to iteratively find the terms
$\varphi_k$.
The coefficients of the eikonal equation (\ref{eikonal}) of the
lowest order in $x$ vanish and the coefficients belonging to
higher orders in $x$ iteratively fix the
$\varphi_k$.
To this end, we expand $B'(x)$ and $B'_\alpha(x)$ at $x=0$ as
\begin{eqnarray}\label{tayB1}
B'(x) &=& \id + DB'|_0 x + O\left(|x|^2\right) \\
B'_\alpha(x) &=& B'_\alpha(0) + DB'_\alpha|_0 x + O\left(|x|^2\right) \;
.\label{tayB2}
\end{eqnarray}
Furthermore we write
\begin{equation}\label{Vnullexpands}
V'_0(x) = V'_{0,q}(x) +  \sum_{k\geq 3}^N W_k(x) +
O\left(|x|^{N+1}\right) \; ,\nonumber
\end{equation}
where $W_k$ denotes a homogeneous polynomial of degree $k$.
The third order equation
\[ -\skp{\nabla \varphi_0}{DB'|_0 x) \nabla\varphi_0} -2 \sum_{\nu =0}^d\lambda_\nu x_\nu
\dnd{\varphi_1}{x_\nu}(x) + W_3 (x) = 0\, ,\qquad x\in\Omega \]
fixes $\varphi_1$ for a given $W_3$, the fourth order
\[ -2\skp{\nabla\varphi_0}{(DB'|_0 x)\nabla\varphi_1} +\sum_{|\alpha|=4}B'_{\alpha}
(\nabla\varphi_0)^\alpha -
2\sum_{\nu =0}^d\lambda_\nu x_\nu\dnd{\varphi_2}{x_\nu}(x) -
\sum_{\nu =0}^d\left(\dnd{\varphi_1}{x_\nu }\right)^2 + W_4(x) = 0\]
is an equation for $\varphi_2$ and the higher orders in
$\varphi$ are inductively given by the higher order parts of the
eikonal equation, which all take the form
\[ \left(\sum_{\nu =1}^d\lambda_\nu x_\nu\dnd{ }{x_\nu }\right)
\varphi_k(x) = v_{k+2}(x)\, , \qquad x\in\Omega\, , \]
with
$v_k=O(|x|^k)$ for $|x|\to 0$.

\section{Expansion of the transformed operator}\label{kap22}

\begin{Def}\label{hypphi}
Let $\psi$ denote any real valued function on $\R^d$.
We introduce an $\ep$-dependent unitary map
\[
U_\ep(\psi)  :
L^2\left(\R^d,dx\right)\rightarrow L^2\left(\R^d,
e^{-2\frac{\psi(\sqrt{\ep}y)}{\ep}}dy\right)=:\mathscr{H}_{\psi}
\]
by
\begin{equation}\label{unit}
(U_\ep(\psi)f)(y) =
\ep^{\frac{d}{4}}e^{\frac{\psi(\sqrt{\ep}y)}{\ep}}f(\sqrt{\ep}y)
\end{equation}
and set for $\hat{H}'_\ep$ as defined in \eqref{hatHstrich}
\begin{equation}\label{defGeps}
\hat{G}_{\ep, \psi} := \tfrac{1}{\ep}\, U_\ep(\psi)
\hat{H}'_{\ep}U^{-1}_\ep(\psi)\; .
\end{equation}
\end{Def}

Then $\hat{G}_{\ep, \psi}$ defines a self adjoint operator on $\mathscr{H}_{\psi}$, whose
domain contains the set of all polynomials $\C[y]$, if $\psi \geq C|x| $ for some $C>0$
and for all large $x$. Choosing in particular
$\psi = \tilde{\varphi}\in\Ce^\infty (\R^d)$  satisfying Hypothesis \ref{tildevarphihyp},
we remark that
for any $M\in\N$, uniform with respect to $\ep\in (0,\ep_0)$,
\begin{equation}\label{tay21}
\| \langle\, . \,\rangle^M\|_{\mathscr{H}_{\tilde{\varphi}}} \leq C_M\, , \qquad
\text{where}\quad
\langle y\rangle:= \sqrt{1+|y|^2}\; .
\end{equation}
In fact by the definition of $\tilde{\varphi}$, for some $A,C_1,C_2>0$,
\begin{equation}\label{tay22}
 \tilde{\varphi}(\sqrt{\ep}y) \geq \begin{cases} C_1 \ep |y|^2\, , \quad \text{for}\;
 |y|\leq \frac{A}{\sqrt{\ep}} \\
                                    C_2 \sqrt{\ep} |y|\, , \quad\text{otherwise}
                                   \end{cases}
\end{equation}
and therefore
\[
 \| \langle\, . \,\rangle^M\|_{\mathscr{H}_{\tilde{\varphi}}} \leq \int_{\R^d} e^{-2C_1|y|^2}
 \langle y\rangle^M \, dy
+ \int_{|y|>\frac{A}{\sqrt{\ep_0}}}
e^{-\frac{C_2 |y|}{\sqrt{\ep_0}}} \langle y\rangle^M\, dy \leq C_M\; .
\]

\begin{prop}\label{TayG}
For $\O_2, \tilde{\varphi}$ as in Hypothesis \ref{tildevarphihyp},
let $\zeta\in\Ce_0^\infty(\R^d)$ be a cut-off-function, such that $\supp \zeta\subset \O_2$
and set $\zeta_\ep(y):= \zeta(\sqrt{\ep}y)$.
Then the operator $\hat{G}_{\ep}:= \hat{G}_{\ep, \tilde{\varphi}}$ defined in \eqref{defGeps}
has an expansion
\begin{equation}\label{Gdach}
 \hat{G}_\ep = \sum_{\tfrac{\N}{2}\ni k\leq N-\frac{1}{2}} \ep^k G_k + R_N\, ,
 \qquad N\in\hNnull\; .
\end{equation}
Here
\begin{equation}\label{G_kbeiGeps}
G_k  =  \left( b_{k} + \sum_{|\alpha|=1}^{2k+2} b_{k, \alpha}
\partial^\alpha\right)   \, ,
\end{equation}
where $b_k$ is a polynomial of degree $2k$, which is even (odd) with respect to $y\mapsto -y$
if $2k$ is even (odd), and
$b_{k, \alpha}$ is a polynomial of degree $2k + 2 - |\alpha|$ , which is even (odd) if $2k -
|\alpha|$ is even (odd).
Moreover there exist constants $C_{N}$ and $\ep_0>0$ such that
for any $\ep\in (0,\ep_0]$ and for any $u,v\in \C[y] \subset \mathscr{H}_{\tilde{\varphi}}$
\begin{equation}\label{restTayG}
 \left| \skpH{u}{\zeta_\ep R_N v} \right| \leq C_{N} \ep^N \sum_{\natop{\alpha\in\N^d}
 {|\alpha|\leq 4N+4}}
 \|y^\alpha u\|_{\mathscr{H}_{\tilde{\varphi}}}
\sum_{\natop{\beta\in\N^d}{|\beta|\leq N}} \||\, . \, |^{2N+2}\partial^\beta
v\|_{\mathscr{H}_{\tilde{\varphi}}}\; .
\end{equation}
\end{prop}

\begin{rem}\label{remG_k}
\ben
\item As a map on $\C [y]$, $G_k\, , \; k\in\frac{\N}{2}$,
raises the degree of a polynomial by $2k$ and preserves (or
changes) the parity with respect to $y\mapsto -y$ according to the
sign $(-1)^{2k}$. This follows at once from the degree and parity
of the polynomials $b_\ell$ in the representation of $G_k$.\\
\item The term of order zero is given more precisely by
\begin{equation}\label{G_0}
G_0 = \Delta_y \varphi_0(y) + \sum_{\nu =1}^d
(2(\partial_{y_\nu}\varphi_0(y))\partial_{y_\nu}) -
\Delta_y + V_1(0) + t_1(0,0)
\end{equation}
This is shown below the proof of Proposition \ref{TayG}.
\een
\end{rem}

\begin{proof}[Proof of Proposition \ref{TayG}]

{\sl Step 1:}\\
We start by analyzing the terms
arising from the potential energy $\hat{V}'_{\ep}$. We have
\[ \tfrac{1}{\ep}U_\ep(\tilde{\varphi}) \hat{V}'_{\ep}U^{-1}_\ep(\tilde{\varphi})
   = \tfrac{1}{\ep}\hat{V}'_{\ep}(\sqrt{\ep}y)
\]
and, using Hypothesis \ref{hyp1},(b), for any $N\in\frac{\N}{2}$ by Taylor expansion of
$V'_\ell(\sqrt{\ep}y)$ for $\ell\in\N, \ell< N$, at $\sqrt{\ep} y =0$, we get for
$N_\ell=2(N-\ell + 1)\in\N$
\begin{align}\label{potepord2}
V'_0(\sqrt{\ep}y) &= \ep \sum_{j=1}^d\lambda_j^2y_j^2 +
\sum_{k = 3}^{N_0-1} \ep^{\frac{k}{2}} D_x^kV'_0|_{x=0}[y]^k +  R_{N,0} (y,\ep)  \\
\ep^\ell  V'_\ell(\sqrt{\ep}y) &=  \sum_{k=0}^{N_\ell-1} \ep^{\ell + \frac{k}{2}} D_x^k
V'_\ell|_{x=0} [y]^k +
R_{N,\ell}(y,\ep) \; , \qquad 1\leq \ell< N \, ,\label{potepord3}
\end{align}
where $D_x^k f|_{x} [y]^k:= D^k_x f|_x (y,\ldots, y)$ and for $0\leq \ell< N$
\begin{equation}\label{Rest1}
 R_{N,\ell}(y,\ep) = \ep^{N+1} \frac{1}{(2N-2\ell )!} \int_0^1 (1-t)^{2(N-\ell)}
 D_x^{2(N-\ell+1)}
V'_\ell |_{x=t\sqrt{\ep}y} [y]^{2(N-\ell+1)} \, dt\; .
\end{equation}
Thus for $u, v\in\C[y]$ and for $0\leq \ell< N$
\begin{equation}\label{Rest2}
 \left|\skpH{u}{\zeta_\ep\frac{1}{\ep} R_{N,\ell}(.,\ep)v}\right| \leq C_{N,\ell} \ep^N
\|u\|_{\mathscr{H}_{\tilde{\varphi}}}
\sum_{\natop{\alpha\in\N^d}{|\alpha|=2(N-\ell+1)}} \| y^\alpha
v\|_{\mathscr{H}_{\tilde{\varphi}}}\; .
\end{equation}
We will need the following notations for $N\in\frac{\N}{2}$:
\begin{equation}\label{N1}
 [N]:= \max \{n\in\N\,|\, n< N+1\}\quad\text{and}\quad [[N]]:= \max \{n\in\N\,|\, n\leq N\}\; .
\end{equation}
Combining the terms in \eqref{potepord2} and \eqref{potepord3} for $0\leq \ell < N$,
leads together with the expansion \eqref{hatVell}
in Hypothesis \ref{hyp1},(b),(i) and the estimates on $R_{[N]}$ given there, to
\begin{align}\label{pot1}
 \tfrac{1}{\ep}U_\ep(\tilde{\varphi}) \hat{V}'_{\ep}U^{-1}_\ep(\tilde{\varphi})(y) &=
 \sum_{j=2}^{2N+1} \ep^{\frac{j}{2}-1} D_x^j V'_0|_{x=0} [y]^j +
 R_{N,0}(y,\ep)\\
&\quad +  \sum_{\ell = 1}^{[N-1]} \left( \sum_{j=0}^{2(N-\ell)+1} \ep^{\frac{j}{2}-1+\ell}
D_x^j V'_\ell |_{x=0} [y]^j
+) R_{N,\ell}(y,\ep)\right)
+ R_{[N]}(y,\ep) \nonumber\\
&=  \sum_{\natop{k\in\hNnull}{k\leq N-\frac{1}{2}}} \ep^k p_k(y) + R_N' (y,\ep)\; ,\nonumber
\end{align}
where
\begin{align}\label{tay45}
p_k(y) &= \sum_{\ell = 0}^{[[k+1]]} D_x^{2(k+1-\ell)} V'_\ell|_{x=0} [y]^{2(k+1-\ell)}
\quad\text{and}\\
R_N'(y,\ep) &= \sum_{\ell=0}^{[N-1]}  R_{N,\ell}(y,\ep) +
\zeta_\ep (y) R_{[N]}(y,\ep)\; .\nonumber
\end{align}
Thus $p_k$ is a polynomial of degree $2k + 2$, which is even (odd) if $2k+2$ is
even (odd) (or if $k$ is integer (half-integer)).\\
It follows from the assumptions given in Hypothesis \ref{hyp1},(b) together with
\eqref{Rest2} that for $u,v\in\C[y]$
\begin{equation}\label{tay41}
 \left|\skpH{u}{\zeta_\ep\frac{1}{\ep} R'_N(.,\ep)v}\right| \leq C_N \ep^N
\|u\|_{\mathscr{H}_{\tilde{\varphi}}}
\sum_{\natop{\alpha\in\N^d}{|\alpha|=2N+2}} \| y^\alpha v\|_{\mathscr{H}_{\tilde{\varphi}}}\; .
\end{equation}

{\sl Step 2:}\\
Now we investigate the coefficients in the expansion of the
kinetic energy $\hat{T}_\ep$ after conjugation with
$U_\ep(\tilde{\varphi})$ on the support of $\zeta_\ep$ and give estimates for the remainder.
By the expansion of $a'_\gamma(x, \ep)$ with respect to $\ep$ following from the assumptions
in Hypothesis \ref{hyp1},(a), we can write for $N\in\frac{\N}{2}$
\begin{align}\label{Tent}
\hat{T}'_\ep &= \sum_{k=0}^{[N-1]} \ep^k \hat{T}_k + \tilde{R}_{[N]} (\ep)\; , \quad
\text{where}\\
\hat{T}_k &:= \sum_{\gamma\in\Gamma_C} a_\gamma^{'(k)} \tau_\gamma \quad\text{and}\quad
\tilde{R}_{[N]}(\ep) =  \sum_{\gamma\in\Gamma_C} R_\gamma^{([N])}(\, . \, ; \ep) \tau_\gamma \; .\label{Tent2}
\end{align}
Using \eqref{astrich}, we get
by Taylor-expansion for $N_k\in\N^*$
\begin{equation}\label{kin2}
 \hat{T}_k g(x) =  \sum_{\natop{\alpha\in\N^d}{|\alpha|< N_k}} \ep^{|\alpha|} B_\alpha^{(k)}(x)
\partial_x^\alpha g|_x  + \tilde{R}'_{N_k} (\ep) g(x)\, , \qquad g\in L^2(\R^d, dx)
\cap \Ce^\infty (\R^d)\,
\end{equation}
where we set
\begin{align}\label{kin1}
B_\alpha^{(k)}(x) &=  \sum_{\eta\in C\Z^d} a_{\ep\eta}^{'(k)} (x) \frac{\eta^\alpha}{\alpha!}=
\frac{1}{\alpha !} \partial_\xi^\alpha t_k (x, i\nabla \tilde{\varphi}(x)) \quad \text{and}\\
\tilde{R}'_{N_k}(\ep) g(x) &= \sum_{\natop{\alpha\in\N^d}{|\alpha|=N_k}}
\sum_{\eta\in C\Z^d} \ep^{|\alpha|} a_{\ep\eta}^{'(k)}(x) \frac{N_k}{\alpha!}\eta^\alpha
\int_0^1 (1-t)^{N_k-1} \partial_x^\alpha g|_{x+t\ep\eta}\, dt \label{kin5}
\end{align}
From Remark \ref{rem1},(b) it follows that $B_\alpha^{(0)} =B'_\alpha = 0$ if $|\alpha|$ is
odd or $|\alpha|=0$.  Inserting
\eqref{kin2} into \eqref{Tent} gives for $\alpha\in \N^d$
\begin{multline}\label{kin3}
 \hat{T}'_\ep g(x) =  \sum_{\natop{|\alpha|=2n, n\in\N}{2\leq |\alpha|< N_0}}
\ep^{|\alpha|}B'_\alpha(x) \partial_x^\alpha g|_x  + \tilde{R}'_{N_0} (\ep)g(x) \\
+   \sum_{k=1}^{[N-1]}  \ep^{k} \left[
\sum_{|\alpha|< N_k} \ep^{|\alpha|}B_\alpha^{(k)}(x)  \partial_x^\alpha g|_x  +
\tilde{R}'_{N_k} (\ep) g(x)\right] + \tilde{R}_{[N]} (\ep)g(x) .
\end{multline}
To analyze the unitary transform of the explicit terms on the right hand side of \eqref{kin3}
we use the following generalized Faa di Bruno formula (see e.g. Hardy \cite{hardy}) for
$g\in\Ce^\infty (\R^d, \R), f\in\Ce^\infty (\R)$ and
$\beta\in\N^d$
\begin{equation}\label{faa}
 \partial^\beta f\circ g = \sum_{n=\min\{1,|\beta|\}}^{|\beta|}
 \sum_{\natop{p\in (\N^d)^n}{\sum_{j=1}^n p_j =\beta, |p_j|\geq 1}}
C_p f^n|_{g} \prod_{j=1}^n \partial^{p_j} g\; .
\end{equation}
\eqref{faa} together with the Leibnitz formula yields for
$f\in\mathscr{H}_{\tilde{\varphi}}$ and  $\alpha, \beta, \beta' \in\N^d$
\begin{multline}\label{faaut}
 U_\ep(\tilde{\varphi})\left[ \ep^k B_\alpha^{(k)}(\ep \partial_x)^\alpha \right]
 U_\ep(\tilde{\varphi})^{-1} f(y) \\
= \ep^{k + |\alpha|}
B_\alpha^{(k)}(\sqrt{\ep}y) \sum_{\beta + \beta' = \alpha}
\sum_{n=\min\{1,|\beta|\}}^{|\beta|}  \sum_{\natop{p\in (\N^d)^n}{\sum_{j=1}^n
p_j =\beta, |p_j|\geq 1}}
C_p \prod_{j=1}^n \left(\partial_x^{p_j} \frac{\tilde{\varphi}}{\ep}\right)|_{\sqrt{\ep}y}
\ep^{-\frac{|\beta'|}{2}} \partial_y^{\beta'} f(y)\; .
\end{multline}
To analyze the right hand side of \eqref{faaut} in detail, we fix $\beta$ and $n$.
Since on the support of $\zeta_\ep$ the phase function $\tilde{\varphi}$ is given by the
asymptotic sum \eqref{varphi},
we have on $\supp \zeta_\ep$
\begin{equation}\label{tay8}
\left|\nabla_x \tilde{\varphi}|_{\sqrt{\ep}y}\right| = O(\sqrt{\ep}) \quad\text{and}\quad
\partial_x^\alpha \tilde{\varphi}|_{\sqrt{\ep}y}=O(1), \,|\alpha|>1,
\qquad (\ep\to 0)\; .
\end{equation}
Thus, for each partition $p$ of $\beta$ of length $n$ (i.e. each
$p=(p_1.\ldots , p_n) \in(\N^d)^n$ with $\sum_{j=1}^n p_j = \beta $), we set
$m_p := m(\beta,n,p) := \# \{p_j\in\N^d\,|\, |p_j| =1\}$.
Then \eqref{tay8} together with \eqref{varphi} yield on the support of $\zeta_\ep$ for
any $N_{\alpha,k}\in\N$ and for $p\in (\N^d)^n$ with $\sum_{j=1}^n p_j =\beta$ and
$|p_j|\geq 1$
\begin{equation}\label{tay2}
\sum_{p, |p_j|\geq 1}  \prod_{j=1}^n\left. \left(\partial_x^{p_j}
\frac{\tilde{\varphi}}{\ep}\right)
\right|_{\sqrt{\ep}y}  =
   \ep^{-n} \sum_{p, |p_j|\geq 1}\left[ \sum_{\ell=0}^{N_{\alpha,k}-1}
\rho_{m_p+\ell}(y) \ep^{\frac{m_p+\ell}{2}} + R''_{N_{\alpha,k}}(y; \ep, m_p) \right]\; ,
\end{equation}
where $\rho_k$ denotes a homogeneous polynomial of degree $k$ and for $u,v\in \C[y]$
\begin{equation}\label{tay11}
 \left|\skpH{u}{\zeta_\ep R''_{N_{\alpha,k}}(.; \ep, m_p)v}\right| \leq C_{N_{\alpha,k}}
 \ep^{\frac{N_{\alpha,k}+ m_p}{2}}
\|u\|_{\mathscr{H}_{\tilde{\varphi}}}
\sum_{\natop{\alpha\in\N^d}{|\alpha|=N_{\alpha,k} + m_p}} \| y^\alpha
v\|_{\mathscr{H}_{\tilde{\varphi}}}\; .
\end{equation}
Note that the polynomials $\rho_k$ depend on $\beta, n, p$ and $\ell$, but are independent of the choice of the truncation
$N, N_k$ and $N_{\alpha,k}$.
For fixed $\beta\in\N^d$ and $n\in\N$, it follows from the definition of $m_p$ that for
any partition $p\in\left(\N^d\right)^n$ of $\beta$ with
length $n$ and $|p_j|\geq 1$
\begin{equation}\label{mp}
\begin{cases} m_p = n\quad\text{if}\quad n=|\beta| \\ (2n-|\beta|)_+ \leq m_p \leq n-1 \quad
\text{if}\quad n<|\beta| \end{cases}\; .
\end{equation}
Thus setting
\begin{equation}\label{M_n}
 M_n := \begin{cases} n-1 \quad \text{for}\quad n<|\beta| \\ n \quad \text{for}\quad
 n=|\beta|\end{cases}\; ,
\end{equation}
the sum over all $p$ on the right hand side of \eqref{tay2} can be substituted by the sum over all $m (=m_p)$ running from $(2n-|\beta|)_+$ to $M_n$.
To expand the right hand side of \eqref{faaut} with respect to $\sqrt{\ep}$, we take Taylor
expansion of
$B_\alpha^{(k)}(\sqrt{\ep}y)$, defined in \eqref{kin1}, at zero up to order $N_{\alpha,k}$,
analog to the expansion of the
potential energy given in \eqref{potepord3}, (we notice that $B_\alpha^{(0)}=0$ if $|\alpha|=0$ or $|\alpha|$ is odd). This yields
for any $k\in\N$ together with \eqref{tay2} and \eqref{faaut} on the support of $\zeta_\ep$
for $\alpha, \beta, \beta'\in\N^d$
\begin{multline}\label{tay4}
U_\ep(\tilde{\varphi})\left[\ep^k B_\alpha^{(k)}
(\ep \partial_x)^\alpha  \right] U_\ep(\tilde{\varphi})^{-1} \\
= \sum_{\beta + \beta' = \alpha}  \sum_{n=\min\{1,|\beta|\}}^{|\beta|}
\sum_{m=(2n-|\beta|)_+}^{M_n}
\sum_{\ell=0}^{N_{\alpha,k}-1}\ep^{k +|\alpha|-n+\frac{m+\ell-|\beta'|}{2}}
 q^{(k)}_{m+\ell}(y) \partial_y^{\beta'}
 + \tilde{\tilde{R}}'_{N_{\alpha,k}}(y; \ep) \; ,
\end{multline}
where
$q_s^{(k)}$ denotes a homogeneous polynomial of degree $s$ (which does not depend on the truncation $N, N_k, N_{\alpha,k}$, but depends on $\alpha, \beta, n, m, \ell$). Moreover  $q^{(0)}_r=0$ for all
$m\leq r <N_{\alpha,k}$, if $|\alpha|=0$ or $|\alpha|$ is odd
and
\begin{multline}\label{tay9}
 \tilde{\tilde{R}}'_{N_{\alpha,k}}(y; \ep)  = \sum_{\beta + \beta' = \alpha}
 \sum_{n=\min\{1,|\beta|\}}^{|\beta|}
\ep^{k+|\alpha|-n-\frac{\beta'}{2}} \sum_{m=(2n-|\beta|)_+}^{M_n} \\
\left[ \sum_{\ell=0}^{N_{\alpha,k}-1}\ep^{\frac{\ell}{2}}
\left(\left.D^\ell B_\alpha^{(k)}\right|_{\sqrt{\ep}y}[y]^{\ell}
R''_{N_{\alpha,k}}(y; \ep, m)\right)\right. \\
\left. + \sum_{\ell=0}^{N_{\alpha,k}-1}\left(\ep^{\frac{m + \ell}{2}}
\rho_{m + \ell}\tilde{R}''_{N_{\alpha,k}}(y,\ep) \right) +  \tilde{R}''_{N_{\alpha,k}}(y,\ep)
R''_{N_{\alpha,k}}(y; \ep, m)\right]
\partial_y^{\beta'}\; .
\end{multline}
where $\tilde{R}''_{N_{\alpha,k}}$ denotes the remaining term in the Taylor expansion of
$B_\alpha^{(k)}$  analog to \eqref{Rest1}
and $R''_{N_{\alpha,k}}$ is the remaining term on the right hand side of \eqref{tay2}.
Analog to \eqref{Rest2} we get for $u,v\in\C[y]$
\begin{equation}\label{tay12}
  \left|\skpH{u}{\zeta_\ep \tilde{R}''_{N_{\alpha,k}}(.,\ep)v}\right| \leq C_{N_{\alpha,k}}
  \ep^{\frac{N_{\alpha,k}}{2}}
\|u\|_{\mathscr{H}_{\tilde{\varphi}}}
\sum_{\natop{\alpha\in\N^d}{|\alpha|=N_{\alpha,k}}} \| y^\alpha
v\|_{\mathscr{H}_{\tilde{\varphi}}}\; .
\end{equation}
Since for fixed $\alpha, \beta, \beta'\in\N^d$ we have
$-n + \frac{m}{2} \geq - \frac{|\beta|}{2}$ for all possible values of $n$ and $m$, it
follows
that $k-1+|\alpha|-n+\frac{m -|\beta'|+\ell}{2} \geq k-1+\frac{|\alpha|}{2}$, thus by
\eqref{tay4} and since $|\alpha|>2$ for $k=0$
the leading order of $\zeta_\ep\frac{1}{\ep}U_\ep(\tilde{\varphi})\hat{T}'_\ep
U_\ep(\tilde{\varphi})^{-1}$ is $\ep^0$.
Moreover these considerations yield by \eqref{tay12}, \eqref{tay9} and \eqref{tay11}
\begin{equation}\label{tay13}
   \left|\skpH{u}{\zeta_\ep \tilde{R}'_{N_{\alpha,k}}(.,\ep)v}\right| \leq C_{N_{\alpha,k}}
\ep^{k + \frac{|\alpha|}{2}-1+\frac{N_{\alpha,k}}{2}}
\sum_{|\alpha|=N_{\alpha,k}}^{2N_{\alpha,k}} \|y^\alpha u\|_{\mathscr{H}_{\tilde{\varphi}}}
\sum_{|\beta'|=0}^{|\alpha|}  \|\partial^{\beta'} v\|_{\mathscr{H}_{\tilde{\varphi}}}\; .
\end{equation}
Combining \eqref{tay4} and \eqref{kin3} leads to
\begin{align}\label{kinepord}
\tfrac{1}{\ep}U_\ep(\tilde{\varphi}) & \hat{T}'_\ep
U_\ep(\tilde{\varphi})^{-1} =: \left( S_1 + S_2 + S_3 + S_4\right)(.,\ep)\, ,
\qquad\text{where for}\;\, v\in{\mathscr H}_{\tilde{\varphi}}\\
S_1(y, \ep) &=  \sum_{k=0}^{[N-1]} \sum_{\natop{\alpha\in\N^d}{|\alpha|< N_k}}
\sum_{\beta + \beta' = \alpha}
\sum_{n=\min\{1,|\beta|\}}^{|\beta|}\ep^{k+|\alpha|-1-n} \nonumber\\
&\hspace{1cm} \times \sum_{m=(2n-|\beta|)_+}^{M_n}
\sum_{\ell=0}^{N_{\alpha,k}}\ep^{\frac{m+\ell-|\beta'|}{2}}
 q^{(k)}_{m+\ell}(y) \partial_y^{\beta'} \label{kine1} \\
S_2(y, \ep) &=\sum_{k=0}^{[N-1]}\sum_{\natop{\alpha\in\N^d}{|\alpha|< N_k}}
\tilde{\tilde{R}}'_{N_{\alpha,k}}(\ep)  \label{kine2} \\
S_3(y, \ep) v(y)&=e^{\frac{\tilde{\varphi}(\sqrt{\ep}y)}{\ep}}\sum_{k=0}^{[N-1]}
\ep^{k-1} \tilde{R}'_{N_k}(\ep)
e^{-\frac{\tilde{\varphi}(x)}{\ep}} \left.v\left(\frac{x}{\sqrt{\ep}}\right)
\right|_{x=\sqrt{\ep}y} \label{kine3}\\
S_4(y, \ep)v(y) &=e^{\frac{\tilde{\varphi}(\sqrt{\ep}y)}{\ep}}  \tilde{R}_{[N]}(\ep)
e^{-\frac{\tilde{\varphi}(x)}{\ep}} \left.v\left(\frac{x}{\sqrt{\ep}}\right)
\right|_{x=\sqrt{\ep}y}\; . \label{kine4}
\end{align}
Recall that $\tilde{R}'_{N_k}$ and $\tilde{R}_{[N]}$ denote operators acting on functions of
$x$.
Choosing $N_{\alpha,k}=2(N+1-k) - |\alpha|$ gives by \eqref{tay13} and \eqref{kine2} for
$u,v\in\C[y]$
\begin{equation}\label{tay33}
   \left|\skpH{u}{\zeta_\ep S_2(\, . \,, \ep) v}\right| \leq C_{N}
\ep^{N}
\sum_{|\alpha|=0}^{4N+4} \|y^\alpha u\|_{\mathscr{H}_{\tilde{\varphi}}}
\sum_{|\beta|=0}^{2N+2}  \|\partial^\beta v\|_{\mathscr{H}_{\tilde{\varphi}}}\; .
\end{equation}
To  estimate the $\mathscr{H}_{\tilde{\varphi}}$-norm of $S_3 v$, we have to analyze the
remainder $R_{N_k}g$ given in \eqref{kin5} in the case
$g(x)=e^{-\frac{\tilde{\varphi}(x)}{\ep}} v(\frac{x}{\sqrt{\ep}})$.
We first remark that by \eqref{tay22}, for $y\in\supp \zeta_\ep$ and for some
$C, C_0>0$, the estimate
\begin{equation}\label{tay23}
 \left| \nabla\tilde{\varphi}({\sqrt{\ep}y + t \ep \eta}) \right| \leq C
\begin{cases} |\sqrt{\ep}y + \ep\eta|\, , &\quad |\eta|\leq \frac{C_0}{\sqrt{\ep}} \\
 1\, , &\quad  |\eta|> \frac{C_0}{\sqrt{\ep}}
\end{cases}
\end{equation}
holds uniformly with respect to $t\in[0,1]$.
Thus for some $C>0$ independent of $t\in[0,1]$, $y\in\supp \zeta_\ep$, by first order Taylor
expansion, we have for any $\eta\in C\Z^d$
\begin{equation}\label{tay24}
 \left|e^{\frac{\tilde{\varphi}(\sqrt{\ep}y)}{\ep}}e^{-\frac{\tilde{\varphi}(\sqrt{\ep}y +
 t\ep\eta)}{\ep}}\right|\leq e^{C|\eta|}\; .
\end{equation}
Moreover by the Leibnitz formula,
\begin{equation}\label{eable}
\partial_x^{\alpha}      e^{-\frac{\tilde{\varphi}(x)}{\ep}} =
e^{-\frac{\tilde{\varphi}(x)}{\ep}}
\left(\nabla_{x} -  \frac{1}{\ep}\nabla\tilde{\varphi}(x)\right)^{\alpha}
\end{equation}
holds. This gives together with \eqref{tay23}
for $C, C_0>0$  and $|\alpha|=N_k$ the estimate
\begin{multline}\label{tay25}
 \left| \int_0^1 (1-t)^{N_k-1} \partial_x^\alpha e^{-\frac{\tilde{\varphi}}{\ep}}
 \left.v\left(\tfrac{.}{\sqrt{\ep}}\right)\right|_{x+t\ep\eta}\, dt \right| \\
\leq  C \sum_{|\tilde{\alpha}|\leq|\alpha|} \int_0^1
\begin{cases} \ep^{-\frac{N_k}{2}}
\left|\tau_{t\sqrt{\ep}\eta}\left( |\, . \,|^{N_k} e^{-\frac{\tilde{\varphi}
(\sqrt{\ep}\, .\,)}{\ep}}\partial^{\tilde{\alpha}}v\right)(y)\right|\, dt  \, ,
\quad\text{if} \;
|\eta|\leq \frac{C_0}{\sqrt{\ep}} \\
\ep^{-N_k} \left|\tau_{t\sqrt{\ep}\eta}\left(e^{-\frac{\tilde{\varphi}
(\sqrt{\ep}\, .\,)}{\ep}}\partial^{\tilde{\alpha}}v\right)(y)\right|\, dt\, ,
\quad\text{if} \;
|\eta|> \frac{C_0}{\sqrt{\ep}}\; .
     \end{cases}
\end{multline}
By \eqref{kin5}, \eqref{kine3}  and \eqref{tay25} it follows that for $v\in\C[y]$
\begin{equation}\label{tay26}
 \|\zeta_\ep S_3(,.\ep) v\|_{\mathscr{H}_{\tilde{\varphi}}}\leq C \sum_{k=0}^{[N-1]}
 \ep^{k-1+N_k}
 \sum_{\natop{\alpha\in\N^d}{|\alpha|=N_k}}\left\|
\zeta_\ep \left\{ A_{1,\alpha}(\, . \, , \ep) + A_{2,\alpha}(\, . \, , \ep) \right\}
\right\|_{\mathscr{H}_{\tilde{\varphi}}}\, ,\end{equation}
where
\begin{align}\label{tay27}
A_{1,\alpha}(y,\ep) := \sum_{|\eta|\leq \frac{C_0}{\sqrt{\ep}}}  \ep^{-\frac{N_k}{2}}
\left|a_{\ep\eta}^{'(k)}(\sqrt{\ep}y)\right| |\eta|^{N_k} \sum_{|\tilde{\alpha}|\leq|\alpha|}
e^{\frac{\tilde{\varphi}(\sqrt{\ep}y)}{\ep}}\int_0^1
\left|\tau_{t\sqrt{\ep}\eta} \left(e^{-\frac{\tilde{\varphi}
(\sqrt{\ep}\, . \,)}{\ep}}|\, . \,|^{N_k}\partial^{\tilde{\alpha}}v\right)\right|\, dt  \\
A_{2,\alpha}(y,\ep) := \sum_{|\eta|> \frac{C_0}{\sqrt{\ep}}}  \ep^{-N_k}
\left|a_{\ep\eta}^{'(k)}(\sqrt{\ep}y)\right| |\eta|^{N_k} \sum_{|\tilde{\alpha}|\leq|\alpha|}
e^{\frac{\tilde{\varphi}(\sqrt{\ep}y)}{\ep}} \int_0^1
\left|\tau_{t\sqrt{\ep}\eta}\left( e^{-\frac{\tilde{\varphi}(\sqrt{\ep}\, . \,)}{\ep}}
\partial^{\tilde{\alpha}}v\right) \right|\, dt\; .
\end{align}
By Cauchy-Schwarz-Inequality, we have for any $c>0$
\begin{multline}
  \|\zeta_\ep A_{2,\alpha}(.,\ep) v\|_{\mathscr{H}_{\tilde{\varphi}}} \leq C  \ep^{-N_k}
  \left\| \left(\sum_{|\eta|> \frac{C_0}{\sqrt{\ep}}}   \left|a_{\ep\eta}^{'(k)}(\sqrt{\ep}\, . \,)\right|^2
  e^{2c|\eta|}\right)^{\frac{1}{2}} \right. \\
 \times \left.
\left( \sum_{|\eta|> \frac{C_0}{\sqrt{\ep}}} e^{-2c|\eta|}  |\eta|^{2N_k}\left[
 \sum_{|\tilde{\alpha}|\leq|\alpha|} e^{\frac{\tilde{\varphi}(\sqrt{\ep}\, . \,)}{\ep}}
 \int_0^1
\left|\tau_{t\sqrt{\ep}\eta}\left( e^{-\frac{\tilde{\varphi}(\sqrt{\ep}\, . \,)}{\ep}}
\partial^{\tilde{\alpha}}v\right) \right|\, dt\right]^2\right)^{\frac{1}{2}}
\right\|_{\mathscr{H}_{\tilde{\varphi}}} \nonumber\\
\leq C  \ep^{-N_k}\sup_{x\in\R^d}\left\|a_{\ep\,.\,}^{'(k)}(x)
e^{c|\, .\,|}\right\|_{\ell_\gamma^2}
 \sum_{|\eta|> \frac{C_0}{\sqrt{\ep}}}   e^{-c|\eta|}  |\eta|^{N_k}e^{C|\eta|} \times \\
 \sum_{|\tilde{\alpha}|\leq|\alpha|}
\int_0^1 \left\|e^{\frac{\tilde{\varphi}(\sqrt{\ep}\, . \,)}{\ep}} \tau_{t\sqrt{\ep}\eta}
\left(e^{-\frac{\tilde{\varphi}(\sqrt{\ep}\, . \,)}{\ep}} \partial^{\tilde{\alpha}}v\right)
\right\|_{\mathscr{H}_{\tilde{\varphi}}} \, dt
\nonumber
\end{multline}
We remark that
\begin{equation}\label{tay40}
\|e^{\frac{\tilde{\varphi}(\sqrt{\ep}\, . \,)}{\ep}} \tau_{t\sqrt{\ep}\eta}
(e^{-\frac{\tilde{\varphi}(\sqrt{\ep}\, . \,)}{\ep}} \partial^{\tilde{\alpha}}v )
\|_{\mathscr{H}_{\tilde{\varphi}}} =  \|\partial^{\tilde{\alpha}}
v \|_{\mathscr{H}_{\tilde{\varphi}}}\; .
\end{equation}
Thus using Hypothesis \ref{hyp1}(a)(iv), we get for each $\alpha\in\N^d, |\alpha|=N_k$
and for any $c'>0$ the estimate
\begin{equation}\label{tay28}
 \|\zeta_\ep A_{2,\alpha}(.,\ep) v\|_{\mathscr{H}_{\tilde{\varphi}}}
\leq  C' e^{-\frac{c'}{\sqrt{\ep}}}\sum_{|\tilde{\alpha}|\leq|\alpha|}
\left\|\partial^{\tilde{\alpha}}v\right\|_{\mathscr{H}_{\tilde{\varphi}}}\; .
\end{equation}
By analog arguments using \eqref{tay25} and \eqref{tay40}, we get
\begin{equation}\label{tay29}
  \|\zeta_\ep A_{1,\alpha}(.,\ep) v\|_{\mathscr{H}_{\tilde{\varphi}}} \leq C
 \ep^{-\frac{N_k}{2}} C''\sum_{|\tilde{\alpha}|\leq|\alpha|}
 \left\||\, . \, |^{N_k} \partial^{\tilde{\alpha}}v\right\|_{\mathscr{H}_{\tilde{\varphi}}}\; .
\end{equation}
Setting $N_k=2(N-k+1)$ and inserting \eqref{tay28} and \eqref{tay29} into \eqref{tay26} gives
\begin{equation}\label{tay30}
 \|\zeta_\ep S_3(,.\ep) v\|_{\mathscr{H}_{\tilde{\varphi}}}\leq C \ep^N
\sum_{|\alpha|\leq 2+2N}\left\||\, . \, |^{N_k}\partial^{\alpha}v
\right\|_{\mathscr{H}_{\tilde{\varphi}}}\; .
\end{equation}
To estimate $S_4$, we use \eqref{Tent2} and the  Cauchy-Schwarz-inequality to get
\begin{align}\label{31}
 \|\zeta_\ep S_4(,.\ep) v\|_{\mathscr{H}_{\tilde{\varphi}}} &=
\left\|\zeta_\ep e^{\frac{\tilde{\varphi}}{\ep}} \sum_{\gamma\in\disk}
R_\gamma^{([N])}(\sqrt{\ep}\, .\,,\ep)e^{\frac{c|\gamma|}{\ep}}
e^{-\frac{c|\gamma|}{\ep}}\tau_\gamma  e^{-\frac{\tilde{\varphi}}{\ep}}v
\right\|_{\mathscr{H}_{\tilde{\varphi}}} \\
&\leq \sup_{x\in\R^d} \left\|R_{(.)}^{([N])}(x,\ep)e^{\frac{c|.|}{\ep}}
\right\|_{\ell^2_\gamma}
\left\| \left(\sum_\gamma \zeta_\ep^2 e^{-\frac{2c|\gamma|}{\ep}}
e^{\frac{2\tilde{\varphi}}{\ep}}\left[\tau_\gamma
e^{-\frac{\tilde{\varphi}}{\ep}}v\right]^2\right)^{\frac{1}{2}}
\right\|_{\mathscr{H}_{\tilde{\varphi}}}\; .\nonumber
\end{align}
\eqref{tay40} and Hypothesis \ref{hyp1}(a)(iv)  yields
\begin{align}\label{tay32}
 \|\zeta_\ep S_4(,.\ep) v\|_{\mathscr{H}_{\tilde{\varphi}}} &\leq \ep^N C \sum_\gamma
 e^{-\frac{c|\gamma|}{\ep}}
\left\| \zeta_\ep e^{\frac{\tilde{\varphi}}{\ep}} \left[\tau_\gamma
e^{-\frac{\tilde{\varphi}}{\ep}}v\right] \right\|_{\mathscr{H}_{\tilde{\varphi}}}\nonumber\\
&\leq C \ep^N \|v \|_{\mathscr{H}_{\tilde{\varphi}}}\; .
\end{align}

{\sl Step 3:}\\
In the last step we are going to combine the terms resulting from
the kinetic and potential energy.
The sum over all $0\leq k \leq [N-1]$ of lhs\eqref{faaut} with $|\beta'|=0$ and
$n=|\beta|(=m_p)$  is given by
\begin{equation}\label{tay5}
\sum_{\natop{\alpha\in\N^d}{|\alpha|=2n, 1\leq n\leq \frac{N_0}{2}}}
\ep^{\frac{|\alpha|}{2}-1} B^{(0)}_\alpha (x) \left(\nabla\tilde{\varphi}|_{x}\right)^\alpha
+ \sum_{k=1}^{[N-1]} \sum_{\natop{\alpha\in\N^d}{|\alpha|< N_k}} \ep^{\frac{|\alpha|}{2}-1 +k}
B_\alpha^{(k)}(x)
\left(\nabla\tilde{\varphi}|_{x}\right)^\alpha  \; ,
\end{equation}
which, by the definition \eqref{kin1} of $B_\alpha^{(k)}$, converges to
$\ep^{-1} t(x,  i\nabla\tilde{\varphi}|_{x})$ as $N\to \infty$.
Since by Hypothesis \ref{tildevarphihyp}, $\tilde{\varphi}$ solves the eikonal
equation \eqref{eikonal} in a neighborhood of $x=0$, it follows that the first sum in
\eqref{tay5} (which for $N_0\to\infty$ converges to
$\ep^{-1}t_0(x, i\nabla\tilde{\varphi}(x))$) cancels
with the potential term $\ep^{-1}V_0(x)$ in each order of $\ep$.
Eliminating the case $\alpha = \beta, n=|\alpha|, k=0$ from the the sum in $S_1$ in
\eqref{kinepord} and eliminating the case
$\ell =0$ in \eqref{tay45}
yields
\begin{align}\label{step3.2}
\hat{G}_\ep
&= \tfrac{1}{\ep}U_\ep(\tilde{\varphi})  \left(\hat{V}'_\ep - V'_0\right)
U_\ep(\tilde{\varphi})^{-1} +
\tfrac{1}{\ep}U_\ep(\tilde{\varphi})  \left(\hat{T}'_\ep - t_0(\, . \, , i\nabla
\tilde{\varphi} (\, . \,))\right)
U_\ep(\tilde{\varphi})^{-1} \\
& =:  A_N + B_N + R_N(\ep) \, . \nonumber
\end{align}

The potential part of order $N$ is
\begin{align}\label{step3.3}
 A_N  = \sum_{\natop{k\in \frac{\N}{2}}{k< N}} \ep^k p'_k \; ,
\end{align}
where
\begin{equation}\label{step3.4}
 p_k(y) = \sum_{\ell = 1}^{[[k+1]]} D_x^{2(k+1-\ell)} V'_\ell|_{x=0} [y]^{2(k+1-\ell)}\; ,
\end{equation}
which is a polynomial of degree $2k - 2$, which is even (odd) if $2k$ is
even (odd) (or if $k$ is integer (half-integer)).
Setting $N_{\alpha,k}= 2(N+1-k) - |\alpha|$ and
$N_k = 2(N + 1 - k)$  in \eqref{kinepord}, the kinetic part of order $N$ in \eqref{step3.2} is
\begin{align}\label{step3.5}
 B_N &= \sum_{\natop{|\alpha|=2n}{1\leq n\leq \frac{2N+2}{2}}} \sum_{\beta + \beta' = \alpha}
\!\sum_{n=\min\{1,|\beta|\}}^{\min \{|\beta|,|\alpha|-1\}}
\!\sum_{m=(2n-|\beta|)_+}^{M_n} \!\sum_{\ell=0}^{2N+2-|\alpha|}
\ep^{|\alpha|-1-n+\frac{m+\ell-|\beta'|}{2}} q^{(0)}_{m+\ell}\partial^{\beta'}  \\
&\! + \sum_{k=1}^{[N-1]}\!\! \sum_{\natop{\alpha\in\N^d}{|\alpha|<2(N+1-k)}}
\!\!\sum_{\beta + \beta' = \alpha}
\sum_{n=\min\{1,|\beta|\}}^{|\beta|}\!\! \!\!\!\!\!\ep^{k-1+|\alpha|-n}\!\!\!\!\!\!\!\!\!
\sum_{m=(2n-|\beta|)_+}^{M_n}\!\!\!
\!\sum_{\ell=0}^{2(N+1-k)-|\alpha|}\!\!\ep^{\frac{m+\ell-|\beta'|}{2}}
 q^{(k)}_{m+\ell} \partial^{\beta'} \, ,
\end{align}
where we recall from \eqref{tay4} that $q^{(k)}_{s}$ is a homogeneous polynomial of degree $s$, which is independent from
the truncation $N$.
First we analyze $B_N$.
In order to get the stated result in \eqref{G_kbeiGeps}, we collect all terms with the
fixed order $r\in\frac{\N}{2}$ in $\ep$. Setting
\begin{equation}\label{ralpha}
r= |\alpha| - 1 - n + \frac{m + \ell - |\beta'|}{2} + k \; ,
\end{equation}
we may rewrite \eqref{step3.5} as
\begin{equation}\label{step3.6}
 B_N = \sum_{\natop{ r\in\frac{\N}{2}}{r<N}}  \sum_{\natop{\beta'\in\N^d}{|\beta'|\leq n_r}}
 \rho_{r, \beta'}
\partial^{\beta'}\; .
\end{equation}
We shall determine $n_r\in\N$ and the properties of $\rho_{r, \beta'}$.
From \eqref{ralpha} we get
\begin{equation}\label{step3.7}
m+\ell = 2r + 2(1+n-|\alpha| - k) + |\beta'| \; .
\end{equation}
By \eqref{step3.5} and \eqref{step3.7}, the polynomial $\rho_{r, \beta'}$ is even (odd) with
respect to
$y\mapsto -y$, if  $2r-|\beta'|$ resp.  is even (odd), since $q^{(k)}_{m+\ell}$ is
homogeneous of order $m+\ell$
and $1+n-|\alpha| - k \in\Z$.
Using \eqref{step3.7} and $\deg q^{(k)}_{m+\ell} = m + \ell$ again, we see that
\begin{equation}\label{step3.8}
 \deg \rho_{r,\beta'} = \max_{(n,\alpha, k)} \left( 2r - 2 ( 1 + n - |\alpha| - k) + |\beta'|
 \right)\; ,
\end{equation}
where $(n, \alpha, k)$ runs through all values occurring in \eqref{step3.5}.
Inspection of
\eqref{step3.5} gives the maximal value
\[ n_{\max} = \begin{cases} \min\{ |\alpha| - |\beta'|, |\alpha| - 1 \} & \quad\text{if}\quad
k=0 \\
               |\alpha| - |\beta'| \quad\text{if}\quad k>0
              \end{cases}
\]
Thus, using \eqref{step3.8},
\begin{equation}\label{step3.9}
\deg \rho_{r,\beta'} = \begin{cases} 2r +2-|\beta'|\, , & |\beta'|>0
\\ 2r\, , &
|\beta'|=0\end{cases} \, .
\end{equation}
It follows from \eqref{step3.8}, that the maximal value $n_r$ of $|\beta'|$ occurring in the
sum on the right hand side of \eqref{step3.6}
is given by  $n_r = 2r+2$.\\
Thus, setting $b_r := \rho_{r, 0} + p_r$ and $b_{r, \beta'} := \rho_{r, \beta'}$, we have by
\eqref{step3.2}, \eqref{step3.3} and \eqref{step3.6} for
any $r\in\tfrac{\N}{2}, r<N$
\[ G_r   =  b_{2r} +
\sum_{\natop{\beta\in\N^d}{1\leq |\beta|\leq 2r+2}} b_{2r+2-|\beta|}\, \partial^\beta \;
, \]
where, by the considerations below \eqref{step3.4} and \eqref{step3.7}, $b_r$ denotes a
polynomial of degree $2r$ which is even
(odd) with respect to $y\mapsto -y$, if $2r$ is even (odd) and $b_{r, \beta}$ denotes a
polynomial of degree $2r + 2 - |\beta|$, which
is even (odd) if $2r - |\beta|$ is even (odd).\\
The estimate \eqref{restTayG} on the remainder $R_N(\ep)$ in \eqref{step3.2} follows at once
from \eqref{tay32}, \eqref{tay30}, \eqref{tay33} and \eqref{tay40}.
\end{proof}

Proof of Remark \ref{remG_k},(b):\\
The kinetic part of the term of order $\ep^0$ results from the terms in \eqref{faaut}, for
which the pair
$(k,|\alpha|)$ takes the values $(0,2), (1,1)$ and $(2,0)$.
These have to be combined with the potential part of this term given by $j=2, \ell=0$ and
$j=0, \ell=1$ respectively
in \eqref{pot1}.
Again, by use of the eikonal equation
\eqref{eikonal}, the terms $V'_{0,q}(x)$ and
$|\nabla\tilde{\varphi}_0(x)|^2$ cancel. Since $B'(x)=\id + o(\id)$, \eqref{G_0} follows by
direct calculation.\\
 \qed\\

We introduce the following formal symbol spaces.
Let for $n\in\N^*$
\begin{align}\label{Seinsdurchn}
\mathcal{K}_{\frac{1}{n}} &:= \left\{ \left.\mu =
\sum_{j\in\frac{\Z}{n}} \mu_j\ep^j \;\right|\; \mu_j\in\C
\quad\textrm{and}\quad
   c_{\mu}:= \inf\{j\,|\,\mu_j\not= 0\}>-\infty \right\}\\
\mathcal{V} &:= \left\{\left. p = \sum_{j\in\frac{\Z}{2}}p_j\ep^j
\;\right|\; p_j\in\C[y]\quad\textrm{and}\quad
   c_p:= \inf\{j\,|\,p_j\not= 0\}>-\infty \right\}\label{K}\, .
\end{align}
Defining addition component-by-component and multiplication by the Cauchy
product,
$\mathcal{K}_{\frac{1}{n}}$ becomes a field of formal Laurent
series with final principal part and $\mathcal{V}$ is a vector
space over $\mathcal{K}_{\frac{1}{2}}$. We can associate to
$\hat{G}_\ep$ a well defined linear operator $G$ on ${\mathcal V}$ by
setting, for ${\mathcal V}\ni p=\sum_{\natop{j\geq
k}{j\in\frac{\Z}{2}}}\ep^j p_j$,
\begin{equation}\label{GinV}
Gp(y) = \sum_{j\geq k}\ep^j \sum_{r\in\frac{\N}{2}}\ep^r G_r p_j(y) =
 \sum_{j+r=\ell\geq k} \ep^\ell G_r p_j(y) \in{\mathcal V} \; .
 \end{equation}
We denote the set of linear operators on $\mathcal{V}$ by $\mathcal{L}(\mathcal{V})$.

We shall define a sesquilinear form on $\mathcal{V}$ with values
in $\mathcal{K}_{\frac{1}{2}}$ (where complex conjugation is
understood component-by-component), which is formally given by
\begin{equation}
\skpk{p}{q} = \int_{\R^d} \overline{p(\ep,y)}q(\ep,y)
e^{-2\frac{\tilde{\varphi}(\sqrt{\ep}y)}{\ep}}\, dy
\, .
\end{equation}
To this end, using \eqref{varphi}, we define real polynomials $\o_k\in\R[y]$ by $\o_0:= 1$ and
\begin{equation}\label{eminus2varphi}
e^{-2\frac{\varphi (\sqrt{\ep}y)}{\ep}} =:
e^{-\sum_{\nu =1}^d\lambda_\nu y_\nu^2}\left(
\sum_{\natop{k\in\hNnull}{k<N}}\ep^k \o_k(y) + \tilde{R}_N(\ep,y) \right)\, .
\end{equation}
Then
\begin{equation}\label{omegaj}
\o_j(y) = \sum_{\ell=1}^{2j}\sum_{\natop{k_1+\ldots + k_\ell = j}{k_i\in\hN}}
\frac{(- 2)^\ell}{\ell!}\varphi_{2k_1}(y) \ldots \varphi_{2k_\ell}(y) \, ,
\end{equation}
where the summands are homogeneous polynomials of degree $2j+2\ell$ with
parity $(-1)^{2j}$ and
\begin{equation}\label{omegajrest}
 |\tilde{R}_N(\ep,y)|  = O\left(\ep^{N} \langle y\rangle^{6N}\right)\; .
\end{equation}

\begin{Def}\label{defskpk}
For $p = \sum_{j\in\frac{\Z}{2}}p_j\ep^j $ and $q =
\sum_{j\in\frac{\Z}{2}}q_j\ep^j$ in $\mathcal{V}$ we define the
sesquilinear form \linebreak
$\skpk{.}{.}:\mathcal{V}\times\mathcal{V}\to\mathcal{K}_{\frac{1}{2}}$
by
\begin{equation}\label{skpk}
\skpk{p}{q} := \sum_{m\in\frac{\Z}{2}}\ep^{m}\sum_{j+k+\ell=m}\;
\int\limits_{\R^d}\overline{p_j(y)}q _k(y)\o_\ell(y) e^{-\sum_{\nu =1}^d
\lambda_\nu y_\nu^2}\, dy\, .
\end{equation}
\end{Def}

Note that $\skpk{p}{q}$ depends only on the Taylor expansion of
$\varphi$ at $0$.

\begin{Lem}
The sesquilinear form defined in (\ref{skpk}) is non-degenerate,
i.e,
\begin{equation}
\skpk{p}{q} = 0\quad\textrm{for all}\quad
p\in\mathcal{V}\quad\textrm{implies}\quad q=0 \, .
\end{equation}
\end{Lem}

\begin{proof}
If $q\not= 0$ we have $q=\sum_{j\geq k}q_j\ep^j,\,k,j\in\hZ$ with $q_k\neq 0$ for
some $k$. For $p:=\ep^k q_k$, the lowest order of the
sesquilinear form is
given by
\[ \ep^{2k} \int_{\R^d}|q_k|^2e^{-\sum_{\nu =1}^d\lambda_\nu y_\nu^2}\, dy >0 \, .\]
Since all other combinations lead to higher orders in $\ep$, this
term can not
be cancelled.
\end{proof}

\begin{prop}\label{Gsymm}
Let $G$ be the operator \eqref{GinV} on $\mathcal{V}$ induced by
$\hat{G}_{\ep}$ defined in \eqref{defGeps} and
let $\skpk{.}{.}$ be the non-degenerate sesquilinear form introduced in Definition
\ref{defskpk}.
Then for all $p,q\in\mathcal{V}$
\[ \skpk{p}{Gq} = \skpk{Gp}{q} \; .\]
\end{prop}

\begin{proof}
We will consider $Y=\C[y]$ as a subset of the form domain of $\hat{G}_{\ep}$ for
$\ep >0$. This can canonically be identified with a subset
$\mathcal{Y}$ of $\mathcal{V}$.
By the linearity of $\skpk{\, . \,}{\, . \,}$, it is sufficient to
prove the proposition for
$p,q\in\mathcal{Y}$.\\
We need the following lemma:

\begin{Lem}\label{Jlem}
Let $p,q\in Y$, then the function
\begin{equation}\label{J1}
(0,\sqrt{\ep_0}) \ni \sqrt{\ep} \mapsto \left\langle p, \hat{G} q\right\rangle (\sqrt{\ep}) :=
\skpH{p}{\hat{G}_\ep q}
\end{equation}
has an asymptotic expansion at $\sqrt{\ep}=0$. In particular, for any $N\in\frac{\N}{2}$,
we have
\begin{equation}\label{J2}
\langle p, \hat{G} q\rangle (\sqrt{\ep}) =
\sum_{\natop{\ell,k\in\frac{\N}{2}}{\ell>0, k+\ell<N}} \ep^{k+\ell} \int_{\R^d}\bar{p}(y)
G_k q(y) e^{-\sum_{\nu=1}^d \lambda_\nu y_\nu^2}\o_\ell (y)\, dy + O\left( \ep^N\right)
\end{equation}
where $G_k$ is given in Proposition \ref{TayG} and the polynomials $\o_\ell$ are
defined in \eqref{eminus2varphi}.
\end{Lem}

\begin{proof}[Proof of Lemma \ref{Jlem}]

{\sl Step 1:}\\
We will show that, for $\zeta_\ep$ as in Proposition \ref{TayG}, there exists some $C>0$ such
that
\begin{equation}\label{J3}
\langle p, \hat{G} q\rangle (\sqrt{\ep}) = \int_{\R^d} e^{-\frac{\varphi (\sqrt{\ep}y)}{\ep}}
\zeta_\ep(y) \bar{p}(y) \hat{G}_\ep q(y) \, dy + O\left( e^{-\frac{C}{\ep}} \right)\; .
\end{equation}
In fact, we can write
\begin{align}
\left\langle p, \hat{G} q\right\rangle (\sqrt{\ep}) &= J_1(\sqrt{\ep}) +
J_2(\sqrt{\ep}) \, , \quad\text{where}\\
J_1(\sqrt{\ep}) := \skpH{\zeta_\ep p}{\hat{G}_\ep q} \quad&\text{and}\quad
J_2(\sqrt{\ep}) := \skpH{(1-\zeta_\ep)p}{\hat{G}_\ep q}\; .
\end{align}
Since by Hypothesis \ref{tildevarphihyp} and the definition of $\zeta$ in Proposition
\ref{TayG} we have $\tilde{\varphi}(x) = \varphi (x)$
for $x\in\supp \zeta$, it remains to show that $|J_2(\sqrt{\ep})|=O( e^{-\frac{C}{\ep}} )$.

By the definition of $\hat{G}_\ep$ and with $x=\sqrt{\ep} y$ we can write
\begin{equation}\label{J4}
 J_2(\sqrt{\ep})  = \ep^{-\frac{d}{2}-1} \int_{\R^d} \bar{p}(\tfrac{x}{\sqrt{\ep}})
 e^{-\frac{\tilde{\varphi}(x)}{\ep}}
\hat{H}'_\ep \left( q(\tfrac{x}{\sqrt{\ep}}) e^{-\frac{\tilde{\varphi}(x)}{\ep}}\right)
(1-\zeta(x))\, dx\; .
\end{equation}
By Hypothesis \ref{hyp1} and \eqref{hatHstrich} we have $\hat{H}'_\ep = \hat{T}'_\ep +
\hat{V}'_\ep$, where
$\hat{T}'_\ep$ is bounded (see Remark \ref{rem1},(c)) and $\hat{V}'_\ep$ is a polynomially
bounded multiplication operator
on $L^2(\R^d)$, thus
\begin{equation}\label{J5}
 (1-\zeta(x)) \hat{H}'_\ep \left( q(\tfrac{x}{\sqrt{\ep}})
 e^{-\frac{\tilde{\varphi}(x)}{\ep}}\right) =: u_\ep(x) \in L^2(\R^d)\; ,
\end{equation}
where $\|u_\ep\|_{L^2} = O(\ep^{-m})$ for some $m>0$ depending on the dimension $d$.
We therefore have by
\eqref{tay22} for some $C>0$, using the Cauchy-Schwarz-inequality in \eqref{J4},
\begin{equation}\label{J6}
 | J_2(\sqrt{\ep})| \leq \ep^{-\frac{d}{2}-1}  \| u_\ep\|_{L^2} \left( \int_{|x|>\eta}
 e^{-\frac{\tilde{C}|x|}{\ep}}
|p(\tfrac{x}{\sqrt{\ep}})|^2 \, dx\right)^{\frac{1}{2}} = O\left(e^{-\frac{C}{\ep}}\right)\; .
\end{equation}

{\sl Step 2:}\\
We will show that for all $N\in\N$
\begin{equation}\label{J7}
\int_{\R^d} e^{-\frac{2\varphi (\sqrt{\ep}y)}{\ep}} \zeta_\ep(y) \bar{p}(y) \hat{G}_\ep
q(y) \, dy =
\!\!\sum_{\natop{\ell,k\in\frac{\N}{2}}{\ell>0, \ell + k <N}} \!\!\!\ep^{\ell + k}
\int_{\R^d}
e^{-\sum_{\nu=1}^d \lambda_\nu y_\nu^2}
\zeta_\ep(y) \o_\ell (y)\bar{p}(y) G_k q(y) \, dy + O(\ep^N)\, ,
\end{equation}
which together with \eqref{J3} shows \eqref{J2}.

By Proposition \ref{TayG}
\begin{equation}\label{J8}
lhs\eqref{J7} = \sum_{\natop{k\in\frac{\N}{2}}{0\leq k <N}} \ep^{k} \int_{\R^d}
e^{-\frac{2\varphi (\sqrt{\ep}y)}{\ep}}
\zeta_\ep(y)\bar{p}(y) G_k q(y) \, dy +O(\ep^N)\; .
\end{equation}

By the definition of $\o_\ell$ in \eqref{eminus2varphi} together with \eqref{omegajrest}, it
follows by the expansion of the exponential function that
\begin{equation}\label{J10}
\zeta_\ep (y) e^{-\frac{2\varphi (\sqrt{\ep}y)}{\ep}} =
\zeta_\ep (y) e^{-\sum_{\nu =1}^d\lambda_\nu y_\nu^2}\left(1+
\sum_{\natop{k\in\hN}{k< N}}\ep^k \o_k(y)+ \hat{R}_N(\ep, y)\right) \, ,
\end{equation}
where for some $\hat{C}_N>0$
\begin{equation}\label{J11}
|\hat{R}_N(\ep, y)| \leq \ep^N \hat{C}_N \langle y \rangle^{6N}\; .
\end{equation}
Inserting \eqref{J10} into \eqref{J8}, using \eqref{J11} and Remark \ref{remG_k}, gives
\eqref{J7}.
\end{proof}

We come back to the proof of Proposition \ref{Gsymm}.\\
In order to use the symmetry of $\hat{G}_{\ep}$ on
$\mathscr{H}_{\tilde{\varphi}}$, we use the function $\langle p, \hat{G} q\rangle$ on
$(0,\sqrt{\ep_0})$
defined in \eqref{J1}.
By Lemma \ref{Jlem} it has an asymptotic expansion at $\sqrt{\ep}=0$, which induces a mapping
\[
\left\langle\, .\,, \hat{G}\, .\,\right\rangle : Y\times Y \to \mathcal{K}_{\frac{1}{2}}\; .\]
By \eqref{GinV},  \eqref{skpk} and \eqref{J2} we see that this function coincides with the
quadratic form $\skpk{\,\cdot\,}{G\,\cdot\,}$
restricted to $\mathcal{Y}\times\mathcal{Y}$. Using the definition \eqref{J1}, we see that
the diagram\\
\begin{center}
$\begin{CD}
Y\times Y   @>\skpH{.}{\hat{G}_\ep.}>>  \mathcal{F}\\
@VV\id V                                @VVTV        \\
\mathcal{Y}\times\mathcal{Y}  @>\skpk{.}{G.}>> \mathcal{K}_{\frac{1}{2}}\\
\end{CD}$\\
\end{center}
is commutative. Here $\mathcal{F}$ denotes the set of functions on $(0,\sqrt{\ep_0})$, which
possess an asymptotic expansion in integer
powers of $\sqrt{\ep}$ and $T$ denotes asymptotic expansion.

Since for $\skpH{\hat{G}_\ep\, .\,}{\,.\,}$ and $\skpk{G\,.}{.}$ we have an analog
commutative diagram,
the proposition is traced back to the symmetry of $\hat{G}_\ep$ on
$\mathscr{H}_{\tilde{\varphi}}$.
\end{proof}

\section{Construction of formal asymptotic expansions}\label{harmapprox}

In this section we construct formal asymptotic expansions for the
eigenfunctions and eigenvalues of $\hat{H}_\ep$.\\
First we recall that the operator $\ep G_0$ on
$\mathscr{H}_{\varphi}$, given in \eqref{G_0}, is unitary
equivalent to the harmonic oscillator
$\hat{H}'^0_q$ defined in \eqref{Hharmonisch},
where the unitary transformation $U_\ep(\varphi_0)$ is defined in
\eqref{unit}. Therefore the spectrum of $G_0$ is given
by
\begin{equation}\label{sigmaKj}
\sigma(G_0) = \left\{\left. e_{\alpha} =
\sum_{\nu =1}^d\left(\lambda_\nu(2\alpha_\nu +1)\right) +
 V_1(0) + t_1(0,0)\; \right|\;
\alpha \in \N^d\right\}\:.
\end{equation}
The eigenfunctions of $\hat{H}'^0_q$ are the functions $g_{\alpha}$
defined in \eqref{gnj} with $\varphi_0$ introduced in
\eqref{phi0}, thus the $\Hi_{\tilde{\varphi}}$-normalized eigenfunctions of $\ep
G_0$ are given by
\begin{equation}\label{Gnullh}
\left(U_\ep(\varphi_0)g_{\alpha}\right)(y) =  h_{\alpha}(y) \, ,\qquad
G_0 h_\alpha = e_\alpha h_\alpha
\end{equation}
where $h_\alpha$ denotes a product of
Hermite polynomials $h_{\alpha_\nu}\in\R[y_\nu]$. Since $h_k(-x)=
(-1)^{k }h_k(x)$ for any $k\in\N$, it follows that $h_{\alpha}$ is even
(respectively odd), if $|\alpha|$ is even (resp. odd).\\
In order to get an expression for the resolvent of the full
operator $G$ on ${\mathcal V}$, we notice that for
$z\notin\sigma(G_0)$ the resolvent $R_0(z)=(G_0 - z)^{-1}$ is well
defined on polynomials and hence on $\mathcal{V}$.

\begin{Lem}\label{R(z)lemma}
Let $z\notin\sigma(G_0)$ and $p,q\in\mathcal{V}$. Then
\ben
\item
the inverse of $(G-z) : \mathcal{V}\to\mathcal{V}$ is given by the
formal von Neumann series
\begin{align} \label{R(z)}
R(z) &:=
\sum_{k=0}^{\infty}\left[-R_0(z)\sum_{j\in\hN}\ep^jG_j\right]^k
R_0(z) = - \sum_{j\in\hNnull}\ep^jr_j(z)
\qquad\text{with}\\
r_j &:= \sum_{k=1}^{2j} (-1)^k \sum_{\natop{j= J_1 +\ldots + j_k}{j_1, \ldots
j_k \in\frac{\N^*}{2}}}
\left(\prod_{m=1}^k -R_0 G_{j_m}\right) R_0\; .\label{rj}
\end{align}
\item
\begin{equation} \label{Rsymm}
\skpk{p}{R(z)q} = \skpk{R(\bar{z})p}{q} \; .
\end{equation}
\item For $r_j$ defined in \eqref{rj}
\begin{equation}\label{r_jsymm}
\skpk{p}{r_j(z)q} = \skpk{r_j(\bar{z})p}{q}, \qquad j\in\hNnull \, .
\end{equation}
\een
\end{Lem}

\begin{proof}

(a): Clearly, $R_0(z) $, $G_j$ and $r_j$ are linear operators in $\mathcal{V}$
(i.e. elements of $\mathcal{L}(\mathcal{V})$ (see Remark \ref{remG_k}), thus
$R(z)\in\mathcal{L}(\mathcal{V})$. A short calculation (in the sense of formal power series)
shows that indeed $R(z) = (G-z)^{-1}$.

(b): By (a) and Proposition
\ref{Gsymm}, we can write
\[ \skpk{p}{R(z) q} = \skpk{(G-\bar{z})R(\bar{z}) p}{R(z)q} = \skpk{R(\bar{z})p}{(G-z)R(z)q}
= \skpk{R(\bar{z}) p}{q}\; . \]

(c): This follows directly from the expansion \eqref{R(z)}.
\end{proof}

In the following we will use the resolvent operator $R(z)$ (as a map on $\mathcal{V}$) to
define a spectral projection for $G$ associated to an eigenvalue
of  $G_0$. \\

By \eqref{R(z)}, $R(z)$ is determined on the polynomials and hence
on ${\mathcal V}$ by the action of the operators $r_j(z):{\mathcal
V}\ra {\mathcal V}$ on the Hermite polynomials, which form a basis
in ${\mathcal Y}$ and thus in ${\mathcal V}$.

It follows from Proposition \ref{TayG}, that $G_j$ raises the
degree of each polynomial by $2j$. Thus there exist real numbers
$c_{\alpha\beta}^j$ such that for all $\alpha,\beta\in\N^d,\,
j\in\hNnull$ we have
\begin{equation}\label{Gjhalpha}
G_jh_{\alpha} = \sum_{|\beta| \leq |\alpha| + 2j}c_{\alpha\beta}^j
h_{\beta}\; .
\end{equation}
Combining \eqref{Gjhalpha}, \eqref{rj} and \eqref{Gnullh}, we can
conclude that there exist rational functions
$d_{\alpha\beta}^j(z)$ with poles at most at the elements of the
spectrum of $G_0$ for which
\begin{equation} \label{rjhalpha}
r_j(z)h_{\alpha} = \sum_{|\beta| \leq |\alpha| + 2j}
d_{\alpha\beta}^j(z) h_{\beta}\, .
\end{equation}
Let $E$ be an eigenvalue of $G_0$ with multiplicity $m$ and let
$\Gamma (E)$ be a circle in the complex plane around
$E$, oriented counterclockwise, such that all other eigenvalues of $G_0$ lie outside of it.\\
Since $r_j(z)$ is well defined on
${\mathcal V}$ for each $j\in\frac{\N}{2}$ and depends in a meromorphic way on $z$, we can
define for $p=\sum_{k\geq M}\ep^k p_k\in{\mathcal V}$
\begin{equation}\label{Pih}
\Pi_E p  := \sum_{\natop{k+\ell=j}{\ell\in\hNnull, k\geq M}}\ep^j\frac{1}{2\pi i}
\oint\limits_{\Gamma (E)}r_\ell(z) p_k\,dz \, .
\end{equation}
We denote this operator by
\[ \Pi_E = -\frac{1}{2\pi i}\oint\limits_{\Gamma (E)}(G-z)^{-1}\,dz \, . \]
This is analog to the familiar Riesz projection for operators on a Hilbert space.

\begin{prop}\label{Piprop}
Let $E\in\sigma(G_0)$ with multiplicity $m$.
Then the operator $\Pi_E$ defined in \eqref{Pih} is a symmetric
projection in $\mathcal{V}$ of dimension $m$, which commutes with
$G$.
\end{prop}

\begin{proof}
{\sl Symmetry:}

The symmetry of $\Pi_E$ is a consequence of \eqref{r_jsymm}:
\[ \skpk{p}{\oint_{\Gamma (E)}r_j(z)\, dz q} = -\skpk{\oint_{\Gamma (E)}r_j(z)\, dz p}{q} \, , \]
where the negative sign results from the conjugation of $z$. \\

$\Pi_E^2 = \Pi_E$:\\
Using \eqref{R(z)}, \eqref{Pih} and the resolvent equation, this follows from standard
arguments (see
\cite{erika} or \cite{hesjo} for the computation in the setting of formal power series).

$\rank\Pi_E = m$:\\
We introduce the set
\begin{equation}\label{IE0}
I_{E} := \left\{ \left. \alpha\in\N^d \,\right|\,
G_0 h_\alpha = E h_\alpha \right\} =: \{\alpha^1, \ldots ,
\alpha^m\}
\end{equation}
numbering the $m$ Hermite polynomials with eigenvalue (energy) $E$
for $G_0$. As a consequence of the representation \eqref{R(z)} of $R(z)$
(recall $r_0(z) = R_0(z)$) and of the definition \eqref{Pih} of $\Pi_E$, we
can write for $\alpha\in I_{E}$
\begin{equation} \label{Pihalpha}
\Pi_E h_{\alpha}
=h_{\alpha} + \sum_{j\in\hN}\ep^j p_j
\end{equation}
for some polynomials $p_j\in\C[y]$ of degree less than or equal to
$|\alpha |  + 2j$ (this follows from \eqref{rjhalpha}). Since the
Hermite polynomials form a basis, \eqref{Pihalpha} implies that
the functions $\Pi_E h_{\alpha^k},\, k=1,\ldots m$, are linearly
independent over $\mathcal{K}_{\frac{1}{2}}$. Thus their span has
dimension $m$. It remains to show that this span coincides with
the range of $\Pi_E$, i.e., we have to show that for all
$\beta\in\N_0^d$ there exist
$\mu_{\alpha}\in\mathcal{K}_{\frac{1}{2}},\,\alpha\in I_{E}$, such
that
\begin{equation} \label{Pihbeta2}
\Pi_E h_{\beta} = \sum_{\alpha\in I_{E}}\mu_{\alpha}h_{\alpha} \,
.
\end{equation}
Let $\beta\notin I_{E}$,
then
\begin{equation} \label{Pihbeta}
\Pi_E h_{\beta} =  \sum_{j\in\hN}\ep^j p_j
\end{equation}
for some $p_j\in\C[y]$. Since the Hermite polynomials form a basis
in $\C[y]$, the polynomial $p_{\frac{1}{2}}$ expands to
\begin{equation}\label{peinhalb}
p_{\frac{1}{2}} = \sum_{\alpha\in I_{E}} c_{\alpha}h_{\alpha} +
\sum_{\beta'\notin I_{E}} c_{\beta'}h_{\beta'} \, .
\end{equation}
Applying $\Pi_E$ on both sides of \eqref{Pihbeta} and using
$\Pi_E^2 = \Pi_E$, \eqref{peinhalb} and again \eqref{Pihbeta} for
the second equality, we get
\[
\Pi_E h_{\beta} = \ep^{\frac{1}{2}}\sum_{\alpha\in I_{E}}
c_{\alpha}\Pi_E h_{\alpha} +
  \ep^{\frac{1}{2}}\sum_{\beta'\notin I_{E}} c_{\beta'}\Pi_E h_{\beta'} +
  \sum_{\natop{j\geq 1}
{j\in\hN}} \ep^j \Pi_E p_j
  =  \ep^{\frac{1}{2}}\sum_{\alpha\in I_{E}} c_{\alpha}\Pi_E h_{\alpha} +
   \sum_{\natop{j\geq 1}{j\in\hN}} \ep^j {\tilde p}_j \,
\]
for ${\tilde p}_j\in \C[y]$. Thus by expanding the terms of the
next order we gain the order $\ep^{\frac{1}{2}}$ in the remaining
term and inductively obtain
$\mu_{\alpha}\in\mathcal{K}_{\frac{1}{2}}$
satisfying equation \eqref{Pihbeta2}. The case $\beta\in I_E$ can easily be reduced to this case.\\

$\Pi_E G = G\Pi_E$:\\
This follows from the fact that $G$ commutes with $R(z)$ together
with the definition
\eqref{Pih}.
\end{proof}

The aim of the following construction is to find an orthonormal
basis in $\ran \Pi_E$, such that $G|_{\ran \Pi_E}$ is represented
by a symmetric $m\times m$-Matrix
$M =(M_{ij})$ with $M_{ij}\in\mathcal{K}_{\frac{1}{2}}$.\\
To this end, we set $f_j := \Pi_E
h_{\alpha^j},\, \alpha^j\in I_{E}$. Then equation \eqref{Pihalpha}
and Definition \ref{defskpk} for the sesquilinear form in
$\mathcal{V}$ imply for some $\gamma_k\in\R$
\begin{equation}\label{fifj}
\skpk{f_i}{f_j} = \delta_{ij} + \sum_{k\in\hN}\ep^k\gamma_k\in\mathcal{K}_{\frac{1}{2}}\, ,
\qquad 1\leq i,j\leq m\, ,
\end{equation}
since the Hermite polynomials are orthogonal and the
$g_{\alpha^j}$ are normalized in the $L^2$-norm.
The matrix $F=(F_{ij}):=(\skpk{f_i}{f_j})$ is symmetric,
because the $f_k$ are real functions. Furthermore the elements of the symmetric matrix
$B:= F^{- \frac{1}{2}}$ (given by a binomial series)  are in
$\mathcal{K}_{\frac{1}{2}}$.
Then
\begin{equation}
e:= (e_1, \ldots ,e_m) := (f_1, \ldots ,f_m) B =: f B
\end{equation}
defines an orthonormal  basis $\{ e_1,\ldots e_m\}$ of $\ran
\Pi_E$ (the orthonormalization of \linebreak $\{f_1,\ldots
,f_m\}$).
In this basis, the matrix $M = (M_{ij})$ of $G|_{\ran \Pi_E}$ is
given by
\begin{equation}\label{Mij}
M = e^tGe = Bf^tGfB = BF^GB\, ,
\end{equation}
where $F^G_{kl} := \skpk{f_k}{Gf_l}\in\mathcal{K}_{\frac{1}{2}}$.
Thus $M$ is a finite symmetric matrix with entries in
$\mathcal{K}_{\frac{1}{2}}$. Using Proposition \ref{TayG} and
\ref{Piprop} together with \eqref{Gjhalpha} and \eqref{fifj} and the fact, that
$h_{\alpha^j},\, \alpha^j\in
I_{E}$, are the eigenfunctions of $G_0$ for the eigenvalue $E$, we
can conclude
\begin{equation}
F^G_{ij} = E\delta_{ij} +
\sum_{k\in\hN}\ep^k\mu_k\, ,\quad\text{where}\quad\mu_k\in\R\, .
\end{equation}
It is shown in \cite{erika}, that $\mathcal{K}:=\bigcup_{n\in\N}
\mathcal{K}_{\frac{1}{n}}$ is algebraically closed, thus any
$m\times m$-matrix with entries in $\mathcal{K}$ possesses $m$
eigenvalues in $\mathcal{K}$, counted with their algebraic
multiplicity. By the following theorem, which is proven in the
appendix of \cite{erika} (see also \cite{greklei}), it actually follows that the
eigenvalues
of matrices with entries in the ring $\mathcal{K}_{\frac{1}{n}}$
also lie in $\mathcal{K}_{\frac{1}{n}}$.

\begin{theo}\label{Meigenw}
Let $M$ be a hermitian $m\times m$-matrix with elements in
$\mathcal{K}_{\frac{1}{n}}$ for some $n\in\N$. Then the
eigenvalues $E_1,\ldots E_m$ are in $\mathcal{K}_{\frac{1}{n}}$
with real coefficients, and the highest negative power occurring
in their expansion is bounded by the highest negative power in the
expansions of $M_{ij}$.\\
Furthermore the associated eigenvectors
$u_j\in(\mathcal{K}_{\frac{1}{n}})^m$ can be chosen to be
orthonormal in the natural inner product.
\end{theo}

We can conclude from Theorem \ref{Meigenw} and the special form of
the elements of $M$ defined in \eqref{Mij} that this matrix
possesses $m$ (not necessarily distinct) eigenvalues in
$\mathcal{K}_{\frac{1}{2}}$ of the form
\begin{equation}\label{Ej}
E_j(\ep) = E + \sum_{k\in\hN}\ep^kE_{jk} =
\sum_{k\in\hNnull}\ep^kE_{jk}\, ,\quad j=1, \ldots m
\end{equation}
where $E_{j0} = E$ and the corresponding eigenfunctions are
\begin{equation}\label{psij}
\psi_j(\ep) = \sum_{k\in\hNnull}\ep^k\psi_{jk} \, , \quad\text{where}\quad
\psi_{jk}\in\C[y]\quad\text{with}
\quad \deg \psi_{jk} = \max_{\alpha\in I_E} (|\alpha| + 2k)\, .
\end{equation}
The statement on the degree of $\psi_{jk}$ follows from
\eqref{Pihalpha} and the fact that every eigenfunction can be
written as linear combination
\begin{equation}\label{psilinear}
\psi = \sum_{\alpha\in I_{E}}\lambda_{\alpha}\Pi_E h_{\alpha}
\end{equation}
with coefficients $\lambda_{\alpha}$ without negative powers in
$\sqrt{\ep}$.

Using the parity results in Proposition \ref{TayG} and Remark \ref{remG_k}, we
can prove the next proposition about the absence of half integer
terms in the expansion \eqref{Ej}.

\begin{prop}\label{yhalf}
Let all $\alpha \in I_{E}$ have the same parity (i.e., $|\alpha|$
is either even for all $\alpha \in I_{E}$ or odd for all $\alpha
\in I_{E}$), where $I_{E}$ is defined in \eqref{IE0}. Let $M$
denote the matrix specified in equation \eqref{Mij} and $E_j(\ep)$
its eigenvalues given in (\ref{Ej}). Then $M_{ij}\in\mathcal{K}_1$
and $E_j(\ep)\in\mathcal{K}_1$ for $1\leq i,j\leq m$.
\end{prop}

\begin{proof}
By Theorem \ref{Meigenw} we know that if $M_{ij}\in\mathcal{K}_1$,
the same is true for the eigenvalues $E_j(\ep)$, so it suffices to
prove the proposition for $M_{ij}$.
We will
change notation during this proof to $f_{\alpha}=\Pi_E h_{\alpha}$ and
$F^G_{\alpha \beta}$ for $\alpha, \beta\in I_{E}$. \\
We start by proving that
$\skpk{f_{\alpha}}{f_{\beta}}\in\mathcal{K}_1$. By definition
\eqref{Pih} the coefficients in the power series of $f_{\alpha}$
are given by
\begin{equation}\label{falphaj}
f_{\alpha j} = \frac{1}{2\pi i}\oint_{\Gamma}r_j(z)h_{\alpha}\, dz\, .
\end{equation}
The $r_j(z)$ are determined by $G_{j_\ell}$ and $R_0(z)$ via formula \eqref{rj}, and since
$G_{j_\ell}$
changes the parity of a polynomial in $\C[y]$ by the factor
$(-1)^{2j_\ell},\, j_\ell \in\hN$ (see Remark \ref{remG_k}), we can conclude
that $r_j(z)$ changes the parity by
$(-1)^{2j}$. Using that the parity of $h_{\alpha}$ is
given by $(-1)^{|\alpha|}$, we obtain $(-1)^{|\alpha|+2j}$ as
parity of $f_{\alpha j}$. By Definition \ref{defskpk} we have
\[
\skpk{f_{\alpha}}{f_{\beta}}
 =\sum_{n\in\hZ}\ep^n
   \sum_{\natop{j,k,\ell\in\hZ}{j+k+\ell=n}}\int_{\R^d}f_{\alpha j}(y)f_{\beta k}(y)
   \o_\ell(y)e^{-\sum_{\nu =1}^d\lambda_{\nu}y_{\nu}^2}\, dy \, .
\]
We shall show that for $2n$ odd (and thus for $n$ half-integer),
each summand vanishes. For fixed $j,k,l$ the integral will vanish
if the entire integrand is odd. According to (\ref{omegaj}) the
parity of $\o_\ell$ is $(-1)^{2\ell}$, the scalar product therefore
vanishes if $(|\alpha| + 2j + |\beta| + 2k + 2\ell)$ is odd. Since by
assumption $\alpha$ and $\beta$ have the same parity, $|\alpha| +
|\beta|$ is even. Thus the integral vanishes if $2(j+k+\ell)= 2n$ is odd, which occurs
if $n$ is half-integer. This shows that
$\skpk{f_{\alpha}}{f_{\beta}}\in\mathcal{K}_1$ and the same is
true for $B_{\alpha\beta}$
by definition. \\
It remains to show the same result for $F^G_{\alpha\beta}$ given
by
\[
\skpk{f_{\alpha}}{Gf_{\beta}} = \sum_{n\in\hZ}\ep^n
   \sum_{\natop{j, k, \ell, r\in\hZ}{j+k+\ell+r=n}}\int_{\R^d}f_{\alpha j}(y)G_rf_{\beta k}(y)
   \o_\ell(y)e^{-\sum_{\nu =1}^d\lambda_{\nu}y_{\nu}^2}\, dy \, .
\]
The operator $G_r$ changes the parity by $(-1)^{2r}$ as already
mentioned, so as before the integral vanishes if $j+k+\ell+r=n$ is
half integer.
\end{proof}

We will now return to our original variable $x = \sqrt{\ep}y$.
Substituting it in equation \eqref{psij} and rearranging with
respect to powers in $\sqrt{\ep}$ yields for $d_{jk}:= \deg \psi_{jk}$ and
$N:=\frac{d_{jk}}{2} - k = \max_{\alpha\in I_E}\frac{|\alpha|}{2}$ to
\begin{align*}
\psi_j(y; \ep) &=
\sum_{k\in\hNnull}\ep^k\psi_{jk}\left(\frac{x}{\sqrt{\ep}}\right)=
\sum_{k\in\frac{\N}{2}}\sum_{\natop{\beta\in\N^d}{0\leq |\beta|\leq d_{jk}}}
\ep^{k-\frac{|\beta|}{2}} \rho_{j,k,\beta} x^\beta  \\
&=: \sum_{\natop{\ell\in\hZ}{\ell\geq -N}} \ep^\ell\hat{u}_{j\ell}(x) \, ,
\end{align*}
where (with $\ell = k - \frac{|\beta|}{2}$)
\begin{equation}
\hat{u}_{j\ell}(x) = \sum_{\natop{\beta\in\N^d}{-2\ell \leq |\beta|}}
\rho_{j,\ell + \frac{|\beta |}{2},\beta}
x^\beta
\end{equation}
and we set
\begin{equation}\label{aj}
\hat{u}_j(x; \ep) := \sum_{\natop{\ell\in\hZ}{\ell\geq -N}}
\ep^\ell \hat{u}_{j\ell}(x)\, .
\end{equation}

We denote by $\mathcal{A}$ the set of formal symbols
$\hat{u}_j$ given by a power series as in \eqref{aj} with
arbitrary $N$. Then $\mathcal{A}$ is a vector space over
$\mathcal{K}_{\frac{1}{2}}$, on which
\begin{equation}\label{HaufA}
\left.e^{\frac{\tilde{\varphi} (x)}{\ep}}\hat{H}'_{\ep}e^{-\frac{\tilde{\varphi} (x)}{\ep}}
\right|_{\mathcal{A}}=: H'_{\ep, \mathcal{A}}
\end{equation}
acts as an operator with eigenfunctions $\hat{u}_j$, where
$\hat{H}_{\ep}$ fulfills Hypothesis \ref{hyp1} and
$\tilde{\varphi}$ is constructed in Hypothesis
\ref{tildevarphihyp}.
The following theorem will summarize these results and give a
condition on the absence of half integer terms in the expansion.

\begin{theo}\label{theo45}
Let $\hat{H}_{\ep}$ satisfy Hypothesis \ref{hyp1}, $\hat{H}'_{\ep}$ be
defined in \eqref{hatHstrich},
and $\tilde{\varphi}$ be the real function described in Hypothesis \ref{hypphi}.
Let $E$ be an eigenvalue with multiplicity $m$ of the harmonic
approximation $G_0$ of $\hat{G}_\ep$ given in \eqref{G_0}.
\ben
\item Then the operator $H'_{\ep, \mathcal{A}}$ defined in \eqref{HaufA} has
an orthonormal system of $m$ eigenfunctions $\hat{u}_{j}$ of
the form \eqref{aj} in $\mathcal{A}$, where the lowest order
monomial in $\hat{u}_{j\ell}\in\C[[x]]$ is of degree $\max\{-2\ell,0\}$.\\
The associated eigenvalues are
\begin{equation}\label{theoev}
\ep E_j(\ep) = \ep \left(E + \sum_{k\in\hN}\ep^kE_{jk}\right)\, .
\end{equation}
\item If $|\alpha|$ is even (resp. odd) for all $\alpha\in I_{E}$,
then all half integer (resp. integer) terms in the expansion
\eqref{aj} vanish.
\een
\end{theo}

\begin{proof}

(a): This point is already shown in the discussion succeeding
equation \eqref{aj}.

(b): By equation \eqref{psilinear} and Proposition \ref{yhalf}
together with Theorem \ref{Meigenw} we can write any eigenfunction
$\psi$ as a linear combination of $\Pi_E h_{\alpha}$ with
coefficients in $\mathcal{K}_1$, thus we get explicitly
\[  \psi \left(\frac{x}{\sqrt{\ep}}\right) = \sum_{\alpha\in I_{E}}
    \sum_{\natop{j\in\N_0}{k\in\N_0/2}} \ep^{j+k}
   \lambda_{\alpha j} f_{\alpha k}\left(\frac{x}{\sqrt{\ep}}\right) \, .
\]
As discussed below \eqref{falphaj}, the polynomials $f_{\alpha k}$
are of degree $(|\alpha| + 2k)$ in $y$, thus they have the order
$\ep^{-(k+\frac{|\alpha|}{2})}$ and the parity of $|\alpha| +2k$,
since they consist of monomials of degree $|\alpha|+2k -2\ell$ for
$0\leq 2\ell \leq |\alpha|+2k,\, \ell\in\N$. If we combine
the powers in $\ep$ arising in the sum, we get
$\ep^{j+\ell-\frac{|\alpha|}{2}}$, where $j$ and $\ell$ are both
integer. If $|\alpha|$ is even, the whole exponent is integer, if
it is odd the exponent is half integer. So if one of these
assumptions is true for all $\alpha\in I_{E}$, there remain no
half integer respectively integer terms. Since the transition to
$\hat{u}_j$ is just a reordering, this is also true for
$\hat{u}_j$. The assertion for \eqref{theoev} follows from Proposition \ref{yhalf}.
\end{proof}

\section{Proof of Theorem \ref{theoEjaj}}\label{approxeigen}

We shall now construct the quasimodes of Theorem \ref{theoEjaj}.
For $\hat{u}_j$ given by \eqref{aj}, we can use the Theorem of
Borel with
respect to $x$, to find ${\mathscr C}^\infty$-functions
$\tilde{\hat{u}}_{jl}$ possessing $\hat{u}_{jl}$ as Taylor
series at zero. We define a formal
asymptotic series in a neighborhood $\O'_3$ of $0$ by
\[ \tilde{\hat{u}}_j(x; \ep) := \sum_{\natop{l\in\hZ}{l\geq -N}} \ep^l
\tilde{\hat{u}}_{jl}(x) \; . \] Then
\begin{equation}\label{WKBmitb}
e^{\frac{\tilde{\varphi} (x)}{\ep}}(\hat{H}'_{\ep} - \ep E_j(\ep))
e^{-\frac{\tilde{\varphi} (x)}{\ep}} \tilde{\hat{u}}_{j}(x,\ep
)= b_j(x; \ep)\; ,
\end{equation}
where $b_j(x; \ep) = \sum_{\natop{l\in\hZ}{\ell\geq -N}} \ep^\ell
b_{j\ell}(x)$ has the property, that each $b_{j\ell}\in\Ce^\infty(\Omega'_3)$ vanishes to
infinite order at $x=0$. It remains to show that it is possible to
modify the functions $\tilde{\hat{u}}_{j\ell}$ by uniquely
determined functions $c_{j\ell}$ vanishing at zero to
infinite order such that for the resulting functions
$\tilde{u}_{j\ell} := \tilde{\hat{u}}_{j\ell} - c_{j\ell}$, the
formal series
\begin{equation}\label{formsera}
\tilde{u}_j(x; \ep):= \sum_{\natop{\ell\geq -N}{\ell\in \Z/2}}\ep^\ell
\tilde{u}_{j\ell}(x)
\end{equation}
solves for $x\in\O'_3$ the equation
\begin{equation}\label{nullinomega3}
e^{\frac{\tilde{\varphi} (x)}{\ep}}(\hat{H}'_{\ep} - \ep E_j(\ep))
e^{-\frac{\tilde{\varphi} (x)}{\ep}} \tilde{u}_{j}(x,\ep )=  0\;.
\end{equation}
To this end, we have to show that the equation
\begin{equation}\label{diffc}
e^{\frac{\tilde{\varphi} (x)}{\ep}}(\hat{H}'_{\ep} - \ep E_j(\ep))
e^{-\frac{\tilde{\varphi} (x)}{\ep}} c_{j}(x,\ep )=
b_j(x; \ep)\; .
\end{equation}
has a unique formal power solution
$c_j(x; \ep)\sim \sum_{\ell\geq -N} \ep^\ell c_{j\ell} (x)$ with
coefficients $c_{j\ell}\in \mathscr{C}^\infty(\O'_3)$ vanishing
to infinite order at $x=0$.
By the definition of $\hat{T}'_\ep$ and $\hat{V}'_\ep$ in \eqref{hatTstrich} and
\eqref{hatVstrich}
and the assumptions in Hypothesis \ref{hyp1}, we have (setting $E=E_{j0}$)
\begin{multline*}
e^{\frac{\tilde{\varphi}(x)}{\ep}} \left[ \hat{T}'_\ep + \hat{V}'_\ep
- \ep\,\left(E + \sum_{k\in\N^*/2}\ep^k E_{jk}\right)\right]
e^{-\frac{\tilde{\varphi}(x)}{\ep}}
\sum_{\natop{\ell\geq -N}{\ell\in \Z/2}}\ep^\ell c_{j\ell}(x) \\
=\sum_{\natop{\ell\geq -N}{\ell\in \Z/2}}\ep^\ell \left\{
\sum_{\gamma\in\disk}\left[\sum_{k\in\N} \ep^k a_\gamma^{'(k)}(x; \ep)
e^{\frac{1}{\ep}(\tilde{\varphi}(x)-\tilde{\varphi}(x+\gamma))}c_{j\ell}(x+\gamma; \ep)\right] \right.\\
\left. + \sum_{k\in\N /2}\ep^k \left(V'_k(x) - \ep
E_{jk}\right)c_{j\ell}(x; \ep)\right\}\; .
\end{multline*}
To get the different orders in $\ep$ of the kinetic term, we
expand $a'_\gamma$, $\tilde{\varphi}$ and $c_{j\ell}$ at $x$ and set $\eta
:= \frac{\gamma}{\ep}\in\Z^d$. Taylor expansion gives
\begin{equation}\label{entphi}
 \frac{1}{\ep}\left(\tilde{\varphi}(x)-\tilde{\varphi}(x+\ep\eta)\right) =
-\nabla\tilde{\varphi}(x)\cdot\eta - \frac{\ep}{2}
D^2 \tilde{\varphi}|_x[\eta]^2 -
 \frac{\ep^2}{2}\int_0^1 (1-t)^2 D^3\tilde{\varphi}|_{x+t\ep \eta}[\eta]^3\, dt
\end{equation}
and
\begin{equation}\label{entc}
c_{j\ell}(x+\ep\eta) = c_{j\ell}(x) + \ep\eta\cdot \nabla c_{j\ell}(x) +
\ep^2 \int_0^1(1-t) D^2 c_{j\ell}|_{x+t\ep\eta} [\eta]^2\, dt\; .
\end{equation}
Combining \eqref{entphi} with the
expansion of the exponential function at zero gives
\begin{multline}\label{entephi}
e^{\frac{1}{\ep}(\tilde{\varphi}(x)-\tilde{\varphi}(x+\gamma))} =
e^{-\nabla\tilde{\varphi}(x)\cdot\eta}
 \left(
1-\frac{\ep}{2}D^2 \tilde{\varphi}|_x[\eta]^2 + \frac{\ep^2}{4}
\left(D^2 \tilde{\varphi}|_x[\eta]^2 \right)^2\!\!\!+ O\left(\ep^4\right)\right)\times\\
 \times \left(1 - \frac{\ep^2}{2}\int_0^1 (1-t)^2 D^3\tilde{\varphi}|_{x+t\ep \eta}[\eta]^3
 \, dt + O\left(\ep^4\right)\right) \left( 1+O\left(\ep^3\right)\right)\; .
\end{multline}
The lowest order equation in \eqref{diffc} is that of order $-N$. By the eikonal
equation \eqref{eikonal}, the left hand side of it vanishes and the same argument
applies for the $-N+\frac{1}{2}$ order equation of \eqref{diffc}. The first
non-vanishing term arises from the action of the first order part
of the conjugated operator on $c_{j,-N}(x)$, which is given by
\begin{multline}\label{firstorder}
 \left\{\sum_{\gamma\in\disk}e^{-\frac{1}{\ep}\nabla\tilde{\varphi}(x)\cdot\gamma}
 \left[ a_{\gamma}^{'(0)}(x)
 \left( \frac{1}{\ep} \gamma\cdot\nabla - \frac{1}{2\ep} \skp{\gamma}{ D^2 \tilde{\varphi}|_x
 \gamma}\right) +
a_\gamma^{'(1)}(x)\right]  + V_1(x) - E\right\} c_{j,-N}(x) \\
= b_{j,-N+1} \, .
\end{multline}
This equation takes the form
\begin{equation}\label{allform}
\left(\mathcal{P}(x,\partial_x) + f(x)\right)u(x) = v(x)
\end{equation}
for the differential operator
\begin{equation}\label{mathP}
\mathcal{P} (x,\partial_x) := \sum_{\eta\in\Z^d}
\tilde{a}'_{\eta}(x) e^{-\nabla\tilde{\varphi} (x)\cdot\eta}
   \eta\cdot\nabla \; ,
\end{equation}
which is well defined by the exponential decay of $\tilde{a}'_\eta$ (see Hypothesis
\ref{hyp1}(a)(iv) and \eqref{agammaunep}) and by \eqref{tay8}. The next and all higher order
equations in \eqref{diffc} result from the action of the first order part of the
conjugated operator given in \eqref{firstorder} on the respective
highest order part of $c_j$, which for the $k$-th order is
the term $c_{j,k-1}$. Additionally to the first order
equation, a term is produced by the action of higher orders of the
conjugated operator on lower order parts of $c_j$. Since
these lower order terms are already determined by the preceding
transport equations, this additional part can be treated as an
additional inhomogeneity of \eqref{allform}. Thus all transport
equation take the form \eqref{allform}
with $f, v\in{\mathscr C}^\infty\left(\O'_3\right)$ and $v$
vanishing to infinite order at $x=0$ by the construction of the
formal series \eqref{aj}. The differential operator ${\mathcal P}$
defined in \eqref{mathP} is of the form $\skp{Z}{\nabla}$ for the
vector field $Z(x)=(z_1(x),\ldots,z_d(x))$ given by
\begin{equation}\label{z}
z_\nu (x)= \sum_{\eta\in\Z^d} \tilde{a}'_{\eta}(x)e^{-\nabla\tilde{\varphi} (x)\cdot\eta}
\eta_\nu \;  \; .
\end{equation}
Since $a_\gamma^{'(0)}(x) = a_{-\gamma}^{'(0)}(x)$ (Remark \ref{rem1}(d)) we have
\begin{equation}\label{aeta}
\sum_{\eta\in\Z^d} \tilde{a}'_\eta(x) \eta_\nu = 0
\quad\text{for all}\quad \nu=1,\ldots, d\; ,
\end{equation}
thus by \eqref{z} $x=0$ is a singular point of the
vector field $Z$. In order to linearize at zero, we compute
\begin{align}
\partial_{x_\mu}|_0 z_\nu &= \sum_{\eta\in\Z^d} \left[ (\partial_{x_\mu}|_0 \tilde{a}'_{\eta})
 e^{-\nabla\tilde{\varphi} (0)\cdot\eta} \eta_\nu - \tilde{a}'_{\eta}(0)
 e^{-\nabla\tilde{\varphi} (0)\cdot\eta}
 \partial_{x_\mu}|_0(\skp{\nabla\tilde{\varphi}}{\eta}) \eta_\nu \right] \nonumber\\
 &=\partial_{x_\mu}|_0 \sum_{\eta\in\Z^d}\tilde{a}'_{\eta}\eta_\nu -
\sum_{\eta\in\Z^d}\tilde{a}'_{\eta}(0) \lambda_\mu \eta_\mu \eta\nu \, ,\label{diffz}
  \end{align}
where for the second equation we used that for $x\in\O$ the phase function
$\tilde{\varphi}$ is given
by \eqref{varphi} and thus $\nabla\tilde{\varphi} (0) = 0$ and
$\partial_{x_\mu}|_0\skp{\nabla\tilde{\varphi}}{\eta} =
\lambda_\mu\eta_\mu $.
By \eqref{aeta} the first term on the right hand side of \eqref{diffz}
vanishes and therefore
\[ \partial_{x_\mu}|_0 z_\nu = - \sum_{\eta\in\Z^d} \tilde{a}'_{\eta}(0)
\lambda_\mu\eta_\mu\eta_\nu \; .
\]
By \eqref{Bnumu} and \eqref{tnullstrich} we get
\[ - \sum_{\eta\in\Z^d} \tilde{a}'_{\eta}(0) \eta_\mu\eta_\nu \lambda_\mu = \left\{
\begin{array}{cl} 2\lambda_\mu > 0 & \mbox{for}\quad \nu=\mu \\
0 & \mbox{for}\quad\nu\neq\mu \; .\end{array}\right. \]
Therefore
the linearization of $Z$ at $0$ is $Z_0:=(z_{10},\ldots,z_{d0})$
with $z_{\nu0} (x) = 2\lambda_\nu x_\nu$ and the corresponding
differential operator is given by
\[ {\mathcal P}_0 (x,\partial_x) = \sum_{\nu = 1}^d 2\lambda_\nu x_\nu \partial_{x_\nu}  \]
with $\lambda_\nu > 0$ for $\nu = 1,\ldots,d$. By Dimassi-Sj\"ostrand \cite{dima}
(Proposition 3.5),
the differential equation \eqref{allform} has a unique ${\mathscr C}^\infty$-solution in a
sufficiently small star-shaped neighborhood $\O'_3$, vanishing to infinite order at $x=0$.
This gives the required solution of \eqref{diffc}, and thus defines $\tilde{u}_j$ in
\eqref{formsera} solving \eqref{nullinomega3} in $\O'_3$.

Again by a Borel procedure, but now with respect to $\ep$, we can
find a function $u'_j\in{\mathscr C}^\infty\left(\O'_3\times
[0,\ep_0)\right)$ representing the asymptotic sum
$\tilde{u}_j(x; \ep)$ given in \eqref{formsera}, which we denote by
\begin{equation}
 u'_j(x; \ep) \sim \sum_{\natop{l\in\hZ}{l\geq -N}} \ep^l \tilde{u}_{jl}(x) \; .
\end{equation}
In order to get a function, which is defined on $\R^d\times
[0,\ep_0)$, we multiply with a cut-off function
$k\in\Ce^\infty_0(\R^d)$, with $\supp k\subset\O'_3$ and such that
for some $\O_3$ with $\overline{\O_3}\subset \O'_K3$ we have
$k(x)=1$ for $x\in\O_3$. We denote the resulting function
$u_j\in\Ce_0^\infty\left(\R^d\times [0,\ep_0)\right)$ by
\[
u_j(x; \ep) := k(x) u'_j(x; \ep) \sim k(x)
\sum_{\natop{\ell\in\hZ}{\ell\geq -N}} \ep^\ell \tilde{u}_{j\ell}(x)=:
\sum_{\natop{\ell\in\hZ}{\ell\geq -N}} \ep^\ell u_{j\ell}(x)\; ,
\]
where $u_{j\ell} := k \tilde{u}_{j\ell}$.
Analogously we define a real function $E_j(\ep)$ as an asymptotic
sum
\[
E_j(\ep) \sim E + \sum_{k\in\frac{\N^*}{2}} \ep^k E_{jk}\; .
\]
We have therefore proven (a), the main part of the Theorem \ref{theoEjaj}.

The approximate orthonormality \eqref{ortho} follows from the orhonormality of the expansion
$\hat{u}_j$ given in \eqref{aj} proven
in Theorem \ref{theo45}, combined with a standard estimate of Laplace type
($ \int e^{-\frac{x^2}{\ep}} O(x^\infty)\, dx = O(x^\infty)$). The estimate \eqref{ortho2}
for the restricted approximate eigenfunctions
follows from \eqref{ortho} together with Lemma 3.4 in \cite{kleinro2}.

The statement on the absence of half-integer terms in \eqref{Erep} follows from Proposition
\ref{yhalf}, the
statement on the absence of half integer or integer terms respectively in \eqref{arep} is a
direct consequence of Theorem \ref{theo45}.

To make the step from $\hat{H}'_\ep$ acting on
$\Ce_0^\infty\left(\R^d\right)$ to the operator $H_\ep$ acting on
lattice functions ${\mathcal K}\left(\disk\right)$, we use that
$\mathscr{G}_{x_0}$ is invariant under
the action of $\hat{H}_\ep$ as discussed above \eqref{Gxnull}.
Thus the restriction to the lattice commutes with $\hat{H}_\ep$
and applying the restriction operator
$r_{\mathscr{G}_{x_0}}$ to \eqref{eigenwg} yields
\eqref{eigenwgdisk}.

\end{document}